%% file: main.tex
\documentclass[a4paper,conference]{IEEEtran}

\usepackage{cite}
\input{header}

\title{Combining Deduction~Modulo and \\
       Logics of Fixed-Point Definitions}

\author{
  \IEEEauthorblockN{David Baelde}
  \IEEEauthorblockA{IT University of Copenhagen}
\and
  \IEEEauthorblockN{Gopalan Nadathur}
  \IEEEauthorblockA{University of Minnesota}
}

\begin{document}

\renewenvironment{proof}{\noindent {\it Proof:} }{\hspace*{\fill}~\IEEEQED\par}

\maketitle




\input{abstract}
\input{intro}
\input{munj}
\input{reductions}
\input{norm}
\input{application}

\input{new}

\medskip
\noindent {\bf Acknowledgements.} The authors thank Gilles Dowek,
Benjamin Werner and the reviewers for helpful discussions and
comments. This work has been funded by the NSF grant 
CCF-0917140; opinions, findings, and conclusions or recommendations
expressed in this paper are those of the authors and do not
necessarily reflect the views of the National Science
Foundation. Nadathur and Baelde are currently receiving support
respectively
from the Velux Fonden and grant 10-092309 from the Danish Council for Strategic
Research, Programme Commission on Strategic Growth Technologies.

\bibliographystyle{IEEEtranS}
\bibliography{master}

\newpage

\input{appendixfigs}

\input{proofs}

\input{formalization}

\end{document}

%% file: header.tex
\usepackage[cmex10]{amsmath}
\usepackage{amsthm}

\usepackage{array}

\usepackage{proof}

\usepackage{url}

\usepackage{rotating}
\usepackage{multirow}
\usepackage[all]{xy}

\newcommand{\ignore}[1]{}

\usepackage{txfonts}

\newcommand{\lfp}[1]{\hbox{\em lfp}(#1)}
\newcommand{\gfp}[1]{\hbox{\em gfp}(#1)}

\newcommand{\F}{{\cal F}}

\newcommand{\x}{\vec{x}}

\renewcommand{\t}{\vec{t}}

\renewcommand{\|}{\; | \;}

\newcommand{\mmuNJ}{\mu\mbox{NJ}}
\newcommand{\muNJ}{$\mmuNJ$}
\newcommand{\muNJmodulo}{\muNJ\ modulo}

\newcommand{\defeq}{\stackrel{\triangle}{=}}
\newcommand{\eqdef}{\stackrel{def}{=}}

\newcommand{\ra}{\rightarrow}

\newcommand{\rew}{\rightsquigarrow}
\newcommand{\ie}{{i.e.,}}
\newcommand{\eg}{{e.g.,}}

\newcommand{\pair}[1]{\langle #1 \rangle}
\newcommand{\unit}{\pair{}}

\newcommand{\subst}[1]{[ #1 ]}
\newcommand{\vin}{\hbox{\em in}}
\newcommand{\proj}{\hbox{\em proj}}
\newcommand{\fst}[1]{\hbox{\em proj}_1(#1)}
\newcommand{\snd}[1]{\hbox{\em proj}_2(#1)}
\newcommand{\csu}[1]{\hbox{\em csu}(#1)}
\newcommand{\textcsu}{\hbox{\em csu}}
\newcommand{\textcsus}{\hbox{\em csu}s}
\newcommand{\refl}[1]{\hbox{\em refl}}

\newcommand{\SN}{{\cal SN}}
\newcommand{\M}{{\cal M}}
\newcommand{\C}{{\cal C}}
\newcommand{\X}{{\cal X}}
\newcommand{\Y}{{\cal Y}}
\newcommand{\E}{{\cal E}}
\newcommand{\interp}[2][\E]{|#2|^{#1}}
\renewcommand{\F}[1]{F_{#1}}

\newcommand{\vdashm}{\vdash} 
\newcommand{\elimv}[3]{\delta_{\vee}({#1},{#2},{#3})}
\newcommand{\elimeq}[2]{\delta_{=}({#1},{#2})}
\newcommand{\elimex}[2]{\delta_{\exists}({#1},{#2})}
\newcommand{\elimbot}[1]{\delta_{\bot}({#1})}

\newcommand{\app}{{\;}}
\newcommand{\ap}{{\,}}
\newcommand{\set}[2]{\{\; #1 \;\vert\; #2 \;\}}
\newcommand{\all}[1]{\{\; #1 \;\}}

\newtheorem{definition}{Definition}
\newtheorem{lemma}{Lemma}
\newtheorem{theorem}{Theorem}
\newtheorem{proposition}{Proposition}

\newtheorem{notation}{Notation}

%% file: abstract.tex
\begin{abstract} Inductive and coinductive specifications are widely 
used in formalizing computational systems. Such specifications have a
natural rendition in logics that support fixed-point definitions.
Another useful formalization
device is that of recursive specifications. These specifications are
not directly complemented by fixed-point reasoning techniques and,
correspondingly, do not have to satisfy strong monotonicity
restrictions. We show how to incorporate a rewriting capability into 
logics of fixed-point definitions towards additionally supporting
recursive specifications. In particular, we describe 
a natural deduction calculus that adds a form of
``closed-world'' equality---a key ingredient to supporting fixed-point
definitions---to \emph{deduction modulo}, a framework for extending a
logic with a rewriting layer operating on formulas. We show that our
calculus enjoys strong normalizability when the rewrite system
satisfies general properties and we demonstrate its usefulness
in specifying and reasoning about syntax-based descriptions. The
integration of closed-world equality into deduction modulo leads us
to reconfigure the elimination principle for this form of equality
in a way that, for the first time, resolves issues regarding
the stability of finite proofs under reduction.
\end{abstract}

%% file: intro.tex
\section{Introduction}\label{sec:intro}

\noindent Fixed-point definitions constitute a widely used
specification device in computational settings. The process of 
reasoning about such definitions can be 
formalized within a logic by including a proof rule for introducing
predicates from their definition, and a case analysis rule for
eliminating such predicates in favor of the definitions through
which they might have been derived.
For example, given the following definition of natural numbers 
\[nat~0 \;\defeq\; \top\qquad nat~(s~x) \;\defeq\; nat~x\]
the introduction and elimination rules would respectively build in the
capabilities of recognizing natural numbers and of reasoning by case
analysis over them.
When definitional clauses are positive, they are
guaranteed to admit a fixed point and the logic 
can be proved to be consistent.
Further, least (resp.~greatest) fixed points can be 
characterized by adding an induction (resp.~coinduction) rule
to the logic.
These kinds of treatments have been added
to second-order logic~\cite{mendler91apal,matthes98csl},
type theory~\cite{paulin93tlca} and 
first-order
logics~\cite{schroeder-Heister93lics,mcdowell00tcs,momigliano08corr,tiu04phd}.  


The case analysis rule, which corresponds under the Curry-Howard
isomorphism to pattern matching in computations, is complex in many
formulations of the above ideas, and the (co)induction rules are even
more so.
By identifying and utilizing a suitable notion of equality, it is
possible to give these rules a simple and elegant rendition. 
For example, the two clauses for $nat$ can be transformed
into the following form:
\[nat~x \;\defeq\; x = 0 \vee \exists y.~ x = s~y \wedge nat~y \]
The case analysis rule can then be derived by unfolding a $nat$
hypothesis into its single defining clause and using elimination
rules for disjunction and equality.
However, to obtain the expected behavior, equality elimination
has to internalize aspects of term equality such as disjointness of
constructors; \eg\ the $0$ branch should be closed immediately if the
instantiation of $x$ has the form $s~n$.
The introduction of this separate notion of equality,
which we refer to as \emph{closed-world equality},
has been central to the concise formulation of generic
(co)induction rules~\cite{tiu04phd}.
Further, fixed-point combinators can be introduced
to make the structure of (co)inductive predicates explicit
rather than relying on a side table of definitions. Thus, the
(inductive) definition of natural numbers may simply be rendered as
$\mu\app (\lambda N \lambda x.~ x = 0 \vee \exists y.~ x = s~y \wedge
N\app y)$.
Fixed point combinators simplify and generalize the theory,
notably enabling mutual (co)induction schemes from the natural (co)induction
rules~\cite{baelde08phd,baelde12tocl}.
The logics resulting from this line of work,
which we refer to as logics of fixed-point definitions from now on,
have a simple structure that
is well-adapted to automated and interactive proof-search
\cite{baelde07cade,baelde10ijcar}.
Moreover, they can be combined with features such as generic
quantification that are useful in capturing binding structure to yield
calculi that are well-suited to formalizing the meta-theory of 
computational and logical systems~\cite{gacek09phd,gacek11ic,miller05tocl}.

Logics featuring (co)inductive definitions can be made 
more powerful by adding another genre of definitions: recursive definitions
based on inductive sets.
A motivating context for such definitions is provided by the Tait-style
strong normalizability 
argument~\cite{tait67jsl}, which figures often in the
meta-theory of computational systems. 
For the simply typed $\lambda$-calculus, this argument relies on
a \emph{reducibility} relation specified by the following clauses: 
\begin{eqnarray*}
   {\sl red}\ \iota\ e &\defeq& {\sl sn}\ e \\
   {\sl red}\ (t_1 \rightarrow t_2)\ e &\defeq&
       \forall e'.~ {\sl red}\ t_1\ e'\supset {\sl red}\ t_2\ (e\ e')
\end{eqnarray*}
We assume here that $\iota$ is the sole atomic type and that {\sl sn}
is a predicate that recognizes strong normalizability. The
specification of {\sl red} looks deceptively like a fixed-point
definition. However, treating it as such is problematic because
the second clause in the definition does not satisfy the positivity
condition.
More importantly, the Tait-style argument does not involve reasoning
on $red$ 
like we reason on fixed-point definitions.
Instead of performing case-analysis or induction on $red$,
properties are proved about it
using an (external) induction on types and the clauses
for {\sl red} mainly support an unfolding of the definition once the
structure of a type is  
known~\cite{girard89book}. Generally, recursive definitions are
distinguished by the fact that they embody computations
or rewriting within proofs rather than the case analysis and
speculative rewriting that is characteristic of fixed-point based
reasoning. 

In this paper, we show how to incorporate the capability of recursive 
definitions into logics of fixed-point definitions.
At a technical level, we do this 
by introducing least and greatest fixed points and
the idea of closed-world equality
into {\em deduction modulo}~\cite{dowek03jar},
a framework for extending a logic with
a rewriting layer that operates on formulas and terms. 
This rewriting layer allows for a transparent treatment of recursive
definitions, but a satisfactory encoding of closed-world equality
(and thus fixed-point definitions)
seems outside its reach.
This dichotomy actually highlights the different 
strengths of logics of fixed-point definitions and deduction modulo:
while the former constitute excellent vehicles for dealing with
(co)inductive definitions, the rewriting capability of the latter is
ideally suited for supporting recursive definitions. By extending
deduction modulo with closed-world equality and fixed points, we achieve a
combination of these strengths. This combination also clarifies the
status of our equality: we show that it is compatible with a theory on
terms and is thus richer than a simple ``syntactic'' form of
equality. 

The main technical result of this paper is a strong normalizability
property for our enriched version of deduction modulo. The seminal
work in this context is that of Dowek and Werner~\cite{dowek03jsl},
who provide a proof of strong normalizability for deduction modulo
that is modular with respect to the rewriting system being used. 
In the course of adapting this proof to our setting, we 
rework previous logical treatments of closed-world equality
 in a way that, for the first time, lets us require that proofs be
 finite  without sacrificing their stability under reduction.
For the resulting system, we are able to construct a proof of strong 
normalizability which follows very naturally the intended 
semantics of fixed-point and recursive definitions:
the former are interpreted as a whole using a semantic fixed-point, 
while the latter are interpreted instance by instance.
Regarding the normalization of least and greatest fixed-point constructs,
our work adapts that of Baelde~\cite{baelde12tocl}
from linear to intuitionistic logic.
We use a natural deduction style in presenting our logic that has 
the virtue of facilitating future investigations of connections with
functional programming.  

The rest of the paper is structured as follows.
In Section~\ref{sec:munj},
we motivate and present our logical system. 
Section~\ref{sec:munj_red} describes reductions on proofs.
Section~\ref{sec:norm} provides a proof of strong normalizability that
is modular in the rewrite rules being considered.
We use this result to facilitate recursive
  definitions in Section~\ref{sec:rec-defs} and we illustrate their
  use in formalizing the meta-theory of programming languages.
Section~\ref{sec:conclusion} discusses 
related and future work.

%% file: munj.tex
\section{Deduction Modulo with Fixed-Points and Equality} \label{sec:munj}

We present our extension to deduction modulo in the form of a typing
calculus for appropriately structured proof terms. This gives us a
convenient tool for defining proof reductions and proving strong
normalizability in later sections.


\subsection{Formalizing closed-world equality}

We first provide an intuition into our formalization of the desired
form of equality. The rule for introducing
an equality is the expected one: two terms are equal if they are
congruent modulo the 
operative rewriting relation. Denoting the congruence by $\equiv$,
this rule can simply be
\[ \infer[t \equiv t']{\Gamma \vdash t = t'}{} \]
The novelty is in the elimination rule that must encapsulate the
closed-world interpretation. This can be captured in the form of a case
analysis over all unifiers of the eliminated equality;
   the unifiers that are relevant to consider here would instantiate
   variables of universal strength, called eigenvariables, in the
   terms. 
One formulation of this idea that has been commonly used in the literature
is the following:
\[ \infer{\Gamma \vdash P}{
      \Gamma \vdash t = t' & \set{\Gamma\theta_i \vdash P\theta_i}{
         \theta_i \in \csu{t,t'}} } \] 
The notation $\csu{t,t'}$ is used here to denote a {\it complete set of
  unifiers} for $t$ and $t'$ modulo $\equiv$, \ie\ a set of
unifiers such that every unifier for the two terms is subsumed by a
member of the set. The closed world assumption is expressed
in the fact that $\Gamma \vdash P$ needs to be proved under only these
substitutions.
Note in particular that the set of right premises here is
empty when $t$
and $t'$ are not unifiable, \ie\ have no common instances.

The equality elimination rule could have simply used the
set of {\it all} unifiers for $t$ and $t'$. Basing it on
\textcsus\ instead allows the cardinality of the premise set to be
controlled, typically permitting it to be reduced to
a finite collection from an infinite one. However, a problem with the
way this rule is formulated is that this
property is not stable under substitution.
For example, consider the following derivation in which $x$ and $y$
are variables:
\[ \infer{p\app x, x = y \vdash p\app y}
         {p\app x, x = y \vdash x = y & (p\app x, x = y)[y/x] \vdash
           (p\app y)[y/x]}\]
If we were to apply the substitution $[t_1/x, t_2/y]$ to it,
the branching structure of the derivation would have to be changed to
reflect the nature of a \textcsu\ for $t_1$ and $t_2$; this could well
be an infinite set. A related problem
manifests itself when we need to substitute a proof $\pi$
for an assumption into the derivation.
If we were to work the proof substitution
eagerly through each of the premises in the equality elimination
rule, it would be necessary to modify the structure of $\pi$ to accord
with the term substitution that indexes each of the premise
derivations.
In the context of deduction modulo, the instantiation in $\pi$
can create new opportunities for rewriting formulas.
Since the choice of the ``right'' premise cannot be determined upfront,
the eager propagation of proof substitutions into equality eliminations
can lead to a form of speculative rewriting which,
as we shall see, is problematic when recursive definitions are
included. 

We avoid these problems by formulating equality elimination in a way
that allows for the {\it suspension} of term and 
proof substitutions. Specifically, this rule is 
\[ \infer{\Gamma' \vdash P\theta}{
      \Gamma' \vdash t\theta = t'\theta &
      \Gamma' \vdash \Gamma\theta &
      \set{\Gamma\theta_i \vdash P\theta_i}{\theta_i \in \csu{t,t'}} } \]
Here, $\Gamma' \vdash \Gamma\theta$ means that there is a derivation
of $\Gamma' \vdash Q$ for any $Q \in \Gamma\theta$.
This premise, that introduces a form of {\it cut}, 
allows us to delay the propagation of proof substitutions over the
premises that represent 
the case analysis part of the rule. Notice 
also that we consider \textcsus\ for $t$ and $t'$ and not $t\theta$ and
$t'\theta$ over these premises, \ie\ the application of the
substitution $\theta$ is also suspended. Of course, these
substitutions must eventually be applied. Forcing the application
becomes the task of the reduction rule for equality that also
simultaneously selects the right branch in the case analysis.

Our equality elimination rule also has the pleasing property of
allowing the structure of proofs to be preserved under
substitutions. For example, the proof 
\[ \infer{p\app x, x = y \vdash p\app y}
         {\infer{\vphantom{[]}p\app x, x = y \vdash x = y}{} &
          \infer{\vphantom{[]}p\app x, x = y \vdash p\app x}{} &
          \infer{(p\app x)[y/x] \vdash (p\app y)[y/x]}{}}\]
under the substitution $\theta := [t_1/x,t_2/y]$ becomes 
\[ \infer{\Gamma \vdash p\app t_2}
         {\infer{\Gamma \vdash (x = y)\theta}{} &
          \infer{\Gamma \vdash (p\app x)\theta}{} &
          \infer{(p\app x)[y/x] \vdash (p\app y)[y/x]}{}}\]
where $\Gamma = (\; p \app t_1,\; t_1 = t_2 \;)$.

\subsection{The logic \muNJmodulo} \label{sec:munj_def}


%

The syntax of our formulas is based on a language of typed $\lambda$-terms.
We do not describe this language in detail and assume only that it is
equipped with 
standard notions of variables and substitutions.
We distinguish $o$ as the type of propositions. Term types,
denoted by $\gamma$, are ones that do not contain $o$.
Predicates are expressions of type $\gamma_1\ra\ldots\ra\gamma_n \ra o$.
Both formulas and predicates are denoted by $P$ or $Q$.
We use $p$ or $q$ for predicate variables and $a$ for predicate constants.
Terms are expressions of term types, and shall be denoted
by $t$, $u$ or $v$. We use $x$, $y$ or $z$ for term variables.
All expressions are considered up to $\beta$- and $\eta$-conversion.
In addition to that basic syntactic equality,
we assume a congruence relation $\equiv$.
In Section~\ref{sec:rec-defs}, we will describe conditions on such a
congruence relation that are sufficient for ensuring the consistency of
the logic. 


\begin{definition}
A unifier of $u$ and $v$ is a substitution $\theta$ such that
$u\theta \equiv v\theta$.
A \emph{complete set of unifiers} for $u$ and $v$, written
$\csu{u,v}$ is a set $\all{\theta_i}_i$ of unifiers of $u$ and $v$, 
such that any other unifier of $u$ and $v$ is of the form
$\theta_i\theta'$ for some $i$ and $\theta'$. Note that complete sets
of unifiers may not be unique. However, this ambiguity will be harmless
in our setting.
\end{definition}

\begin{definition} 
Formulas are built as follows:
\begin{eqnarray*}
P &::=& \top \| \bot \| P \supset Q \| P \wedge Q \| P \vee Q
     \| \forall x. P \| \exists x. P \| \\
  &  & t = t'
      \| (\mu\ B\ \t) \| (\nu\ B\ \t) \| (p\ \t) \| (a\ \t)
\end{eqnarray*}
Here, $\wedge$, $\vee$ and $\supset$ are connectives of type $o \ra o
\ra o$, equality has type $\gamma \ra \gamma \ra o$
and quantifiers have type $(\gamma \ra o) \ra o$
for any $\gamma$.
Expressions of the form $a\ap\t$ are called \emph{atomic formulas}.
The least and greatest fixed point combinators $\mu$ and $\nu$
have the type $(\tau \ra \tau) \ra \tau$ for any
$\tau$ of the form $\vec{\gamma} \ra o$.
The first argument for these combinators, denoted by $B$, must have
the form $\lambda p \lambda \x. P$ called a predicate operator.
Every predicate variable occurrence must be within such an operator,
bound by the first abstraction in it. 
An occurrence of $p$ in a formula is \emph{positive} if it is on the 
left of an even number of implications, and it is \emph{negative}
otherwise and $\lambda p \lambda \x. P$ is said to be monotonic
(resp.~antimonotonic) if $p$ occurs only positively (resp.~negatively)
in $P$. We restrict the first argument of fixed-point combinators to
be monotonic operators.
\end{definition}


\begin{figure*}[!t]
\renewcommand{\arraystretch}{2}
\[ \begin{array}{r>{\quad}l}
  \infer[P\equiv Q,\, (\alpha:Q)\in\Gamma]{\Gamma\vdashm \alpha:P}{} &
  \infer[P\equiv \top]{\Gamma \vdashm \unit : P}{}
  \qquad
  \infer{\Gamma\vdashm \elimbot{\pi}:P}{\Gamma\vdashm \pi:\bot}
  \\
  \infer[P\equiv P_1 \supset P_2]{\Gamma\vdashm \lambda\alpha.\pi:P}{
      \Gamma, \alpha:P_1 \vdashm \pi:P_2} &
  \infer{\Gamma\vdashm \pi\app\pi':P}{
      \Gamma\vdashm \pi:Q \supset P & \Gamma\vdashm \pi':Q} \\
  \infer[P\equiv P_1\wedge P_2]{\Gamma\vdashm \pair{\pi_1,\pi_2}:P}{
      \Gamma\vdashm \pi_1:P_1 & \Gamma\vdashm \pi_2:P_2} &
  \infer[P'_i \equiv P_i,\, i \in \{1,2\}]{\Gamma \vdashm \hbox{\em proj}_i(\pi):P'_i}{
      \Gamma\vdashm \pi:P_1 \wedge P_2} \\
  \infer[P\equiv P_1\vee P_2]{\Gamma \vdashm \vin_i(\pi):P}{
      \Gamma \vdashm \pi:P_i}
  &
  \infer{
      \Gamma \vdashm \elimv{\pi}{\alpha.\pi_1}{\beta.\pi_2}:P}{
      \Gamma \vdashm \pi:P_1\vee P_2 &
      \Gamma, \alpha:P_1 \vdashm \pi_1:P &
      \Gamma, \beta:P_2 \vdashm \pi_2:P}
  \\
  \infer[P \equiv \forall x. Q]{\Gamma\vdashm \lambda x.\pi : P}{
      \Gamma\vdashm \pi:Q} &
  \infer[P \equiv Q\subst{t/x}]{
      \Gamma\vdashm \pi\app t : P}{
      \Gamma\vdashm \pi: \forall x. Q} \\
  \infer[P \equiv \exists x. Q]{
      \Gamma\vdashm \pair{t,\pi}:P}{
      \Gamma\vdashm\pi:Q\subst{t/x}} &
  \infer{\Gamma\vdashm \elimex{\pi}{x.\alpha.\pi'}:P}{
           \Gamma\vdashm \pi: \exists x. Q &
           \Gamma,\alpha:Q\vdashm\pi':P}
  \\
  \infer[P \equiv (t=t)]{\Gamma\vdashm \refl{t}:P}{}
  &
  \infer[(\theta'_i)_i \in \csu{u,v}, P \equiv Q\theta]{
      \Gamma' \vdashm
         \elimeq{\Gamma,\theta,\sigma,u,v,Q,\pi}{(\theta'_i.\pi_i)_i}:P}{
      \Gamma' \vdashm \pi : u\theta=v\theta &
      \Gamma' \vdash \sigma : \Gamma\theta &
      ( \Gamma\theta'_i \vdashm \pi_i:Q\theta'_i )_i}
  \\
  \infer[P \equiv \mu\app B\app \t]{
      \Gamma \vdashm \mu(B,\t,\pi) : P}{\Gamma \vdashm \pi:B\app (\mu\app
    B)\app \t}
  &
  \infer[P\equiv S\app \t]{
      \Gamma \vdashm \delta_\mu(\pi,\x.\alpha.\pi'):P}{
      \Gamma \vdashm \pi : \mu\app B\app \t &
      \Gamma, \alpha : B\app S\app \x \vdash \pi' : S\app \x}
  \\
  \infer[P\equiv \nu\app B\app \t]{
      \Gamma \vdashm \nu(\pi,\x.\alpha.\pi') : P}{
      \Gamma \vdashm \pi : S\app \t &
      \Gamma, \alpha : S\app \x \vdashm \pi' : B\app S\app \x}
  &
  \infer[P\equiv B\app(\nu\ap B)\app\t]{
      \Gamma \vdashm \delta_\nu(B,\t,\pi) : P}{
      \Gamma \vdashm \pi : \nu \ap B \ap \t}
\end{array} \]
Variables bound in proof terms are assumed to be new in instances of typing
rules, \ie\ they should not occur free in the base sequent.
Specifically, $\alpha$, $\beta$, $x$ are assumed to be new
in the introduction rules for implication, universal quantification
and greatest fixed-point,
as well as elimination rules for disjunction, existential quantification,
equality and least fixed-point.
\caption{\muNJ: Natural deduction modulo with equality and
   least and greatest fixed points}
\label{fig:munj}
\end{figure*}

We now introduce a language of proof terms, and define type assignment.
The terms and typing rules for all but the equality and fixed point
cases are standard (\eg\ see  
\cite{dowek03jsl}). Following the Curry-Howard correspondence,
(proof-level) types correspond to formulas,
typing derivations correspond to proofs,
and the reduction of proof terms corresponds to proof normalization.
The guidelines determining the form of the new proof terms 
are that all information needed for reduction should be included in
them and that type checking should be easily decidable. The details of
our choices should become clear when we present the typing rules.

\begin{definition}\label{proofterms} 
Proof terms, denoted by $\pi$ and $\rho$, are given by the following
syntax rules: 
\begin{eqnarray*}
\pi &::=& \alpha \| \unit{} \| \elimbot{\pi} \\
    &|& \lambda\alpha.\pi \| (\pi\app\pi') \\
    &|& \pair{\pi,\pi'} \| \fst{\pi} \| \snd{\pi} \\
    &|& \vin_1(\pi) \| \vin_2(\pi) \|
        \elimv{\pi_1}{\alpha.\pi_2}{\beta.\pi_3} \\
    &|& \lambda x.\pi \| (\pi\app t) \\
    &|& \pair{t,\pi} \| \elimex{\pi}{x.\alpha.\pi'} \\
    &|& \refl{t} \|
        \elimeq{\Gamma,\theta,\sigma,u,v,P,\pi}{(\theta'_i.\pi_i)_i} \\
    &|& \mu(B,\t,\pi) \| \delta_\mu(\pi,\x.\alpha.\pi') \\
    &|& \nu(\pi,\alpha.\pi') \|
        \delta_\nu(B,\t,\pi)
\end{eqnarray*}
Here and later,
we use $\alpha$, $\beta$, $\gamma$ to denote proof variables,
and $\sigma$ to denote substitutions for proof variables.
The notation $(\theta'_i.\pi_i)_i$ in the equality elimination construct
stands for a finite, possibly empty, collection of subterms. 
In the expression $\theta.\pi$, all free variables of $\pi$
must be in the range of the substitution $\theta$. 
%
Finally, the notation $x.\pi$ or $\alpha.\pi$ denotes a binding construct,
\ie\ $x$ (resp.~$\alpha$) is bound in $\pi$.
As usual, terms are identified up to a renaming of bound variables,
and renaming is used to avoid capture when propagating a substitution under a
binder.
\end{definition}



Typing judgments are relativized to {\em contexts} that are
assignments of types to finite sets of proof variables. We denote
contexts by $\Gamma$, written perhaps with subscripts and
superscripts.  

\begin{definition}
A proof term $\pi$ has type $P$ under the context $\Gamma$
if $\Gamma\vdashm \pi:P$ is derivable using the rules in 
Figure~\ref{fig:munj}.
We also say that $\Gamma' \vdashm \sigma : \Gamma$ holds
if $\Gamma$ and $\sigma$ have the same domain and
$\Gamma' \vdashm \sigma(\alpha) : \Gamma(\alpha)$
holds for each $\alpha$ in that domain.
\end{definition}

%


\subsection{Expressiveness of the logic} \label{sec:munj_xpl}

The logic \muNJmodulo\ inherits from
logics of fixed-point definitions a simplicity in the treatment of
(co)inductive sets and relations and from deduction modulo the ability
to blend computation and deduction in the course of reasoning. We
illustrate this aspect through a few simple examples here.

Natural numbers may be specified through the following least 
fixed point predicate:
\[ nat \eqdef \mu\app (\lambda N \lambda x.~ x = 0 \vee \exists y.~ x =
s~y \wedge N~y) \]
Specialized for this predicate, the least fixed point rules immediately
give rise to the following standard derived rules:
\[ \infer{\Gamma \vdash nat~0}{} \quad
   \infer{\Gamma \vdash nat~(s~x)}{\Gamma \vdash nat~x} \]
\[ \infer[y \mbox{ new}]{\Gamma \vdash P~x}{
     \Gamma \vdash nat~x &
     \Gamma \vdash P~0 & \Gamma, P~y \vdash P~(s~y)} \]

Having natural numbers, we can easily obtain the rest of Heyting arithmetic.
Addition may be defined as an inductive relation,
but the congruence also allows it to be defined more naturally
as a term-level function,
equipped with the rewrite rules $0+y \rew y$ and $(s~x)+y \rew
s~(x+y)$. Treating it in the latter way allows us to exploit the
standard dichotomy between deduction and computation in deduction
modulo to shorten proofs~\cite{burel07csl}. For example, $(s~0)+(s~0) = s~(s~0)$
can be proved in one step by using the fact that the two terms in the
equation are congruent to each other. 
More general properties about addition defined in this way must be
conditioned by assumptions about the structure of the terms.
For instance, commutativity of addition should be stated as follows:
\[ \forall x \forall y.~ nat~x \supset nat~y \supset x+y = y+x \]
This proposition can be proved by induction on the $nat$ hypotheses,
with the computation of addition being performed implicitly
through the congruence when the structure of the first summand becomes
known. 
Note that we do not have to know how to compute \textcsus{} modulo arithmetic
to build that derivation: all that is needed is the substitutivity 
principle $\forall x \forall y.~ x = y \supset P~x \supset P~y$
which only involves shallow unification.


%


%% file: reductions.tex
\section{Reductions on Proof Terms}\label{sec:munj_red}

\begin{figure*}[!t]
\[ \begin{array}{rcl}
 (\elimeq{\Gamma,\theta',\sigma,u,v,P,\pi}{(\theta''_i.\pi_i)_i})\theta
 &\eqdef&
 \elimeq{\Gamma,\theta'\theta,\sigma\theta,u,v,P,\pi\theta}{
            (\theta''_i.\pi_i)_i}
\\
 (\elimeq{\Gamma,\theta',\sigma',u,v,P,\pi}{(\theta''_i.\pi_i)_i})\sigma
 &\eqdef&
    \elimeq{\Gamma,\theta',\sigma'\sigma,u,v,P,\pi\sigma}{
      (\theta''_i.\pi_i)_i}
\end{array} \]
\[ \begin{array}{rclrcl}
 (\mu(B,\t,\pi))\theta &\eqdef& \mu(B\theta,\t\theta,\pi\theta)
&
 (\delta_\mu(\pi,\x.\alpha.\pi'))\theta
 &\eqdef&
 \delta_\mu(\pi\theta,\x.\alpha.\pi'\theta)
\\
 (\mu(B,\t,\pi))\sigma &\eqdef& \mu(B,\t,\pi\sigma)
&
 (\delta_\mu(\pi,\x.\alpha.\pi'))\sigma
 &\eqdef&
 \delta_\mu(\pi\sigma,\x.\alpha.\pi'\sigma)
\end{array} \]
\[ \begin{array}{rclrcl}
 (\nu(\pi,\x.\alpha.\pi'))\theta &\eqdef& \nu(\pi\theta,\x.\alpha.\pi'\theta)
&
 (\delta_\nu(B,\t,\pi))\theta
 &\eqdef&
 \delta_\nu(B\theta,\t\theta,\pi\theta)
\\
 (\nu(\pi,\x.\alpha.\pi'))\sigma &\eqdef& \nu(\pi\sigma,\x.\alpha.\pi'\sigma)
&
(\delta_\nu(B,\t,\pi))\sigma
 &\eqdef&
\delta_\nu(B,\t,\pi\sigma)
\end{array} \]
\caption{Term and proof-level substitutions into equality, least and greatest fixed-point
  proof terms}
\label{fig:substitution}
\end{figure*}

As usual, we consider reducing proof terms in
which an elimination rule for a logical symbol immediately follows an
introduction rule for the same symbol. Substitutions for both
term-level and proof-level variables play an important role in
describing such reductions. They are defined as usual,
extended as shown on Figure~\ref{fig:substitution} for equality and 
for the least and greatest fixed-point constructs.
Note that substitutions are
suspended over the parts representing case analysis in the equality
elimination rule as discussed earlier. 
The next two lemmas show that this treatment of substitution is
coherent. 

\begin{lemma} \label{lem:termsubstred}
Term-level substitution preserves type assignment:
$\Gamma\vdashm \pi:P$ implies $\Gamma\theta\vdashm \pi\theta:P\theta$.
\end{lemma}

\begin{proof}
This is easily checked by induction on the typing derivation. An
interesting case is that of equality elimination.
Consider the following derivation:
\[ \infer[P\equiv P'\theta']{
     \Gamma' \vdashm
         \elimeq{\Gamma,\theta',\sigma,u,v,P',\pi}{(\theta''_i.\pi_i)_i} : P}{
             \Gamma' \vdashm \pi : u\theta' = v\theta' &
             \Gamma' \vdash \sigma : \Gamma\theta' &
             ( \Gamma\theta''_i \vdashm \pi_i : P'\theta''_i )_i
             } \]
By the induction hypothesis,
$\Gamma'\theta \vdash \pi\theta : u\theta'\theta=v\theta'\theta$
and $\Gamma'\theta \vdash \sigma\theta : \Gamma\theta'\theta$ have
derivations. From these we build the derivation
\[ \infer{
     \Gamma'\theta \vdashm
        \elimeq{\Gamma,\theta'\theta,\sigma\theta,u,v,P',\pi\theta}{
                (\theta''_i.\pi_i)_i} : P\theta}{
            \Gamma'\theta
              \vdashm \pi\theta : u\theta'\theta=v\theta'\theta &
            \Gamma'\theta \vdashm \sigma\theta : \Gamma\theta'\theta &
            ( \Gamma\theta''_i \vdashm \pi_i : P'\theta''_i )_i
            } \]
\end{proof}

\begin{lemma}\label{lem:proofsubstred}
If $\Gamma \vdashm \pi : P$ and $\Gamma' \vdashm \sigma : \Gamma$
then $\Gamma' \vdash \pi\sigma : P$. 
\end{lemma}

\begin{proof}
This is shown also by induction on the typing derivation. An
interesting case, again, is that of equality elimination.
Consider the following derivation:
\[ \infer[P\equiv P'\theta]{
     \Gamma
     \vdashm \elimeq{\Gamma'',\theta,\sigma',u,v,P',\pi}{(\theta'_i.\pi_i)_i}:P}{
       \Gamma \vdashm \pi : u\theta = v\theta &
       \Gamma \vdashm \sigma' : \Gamma''\theta &
       ( \Gamma''\theta'_i \vdashm \pi_i : P'\theta'_i )_i
       } \]
By the induction hypothesis,
$\Gamma' \vdash \pi\sigma : u\theta=v\theta$
and
$\Gamma' \vdash \sigma'\sigma : \Gamma''\theta$ have derivations.
From this we build the derivation
\[ \infer{
     \Gamma' \vdashm
        \elimeq{\Gamma'',\theta,\sigma'\sigma,u,v,P',\pi\sigma}{
                (\theta'_i.\pi_i)_i} : P}{
            \Gamma' \vdashm \pi\sigma : u\theta=v\theta &
            \Gamma' \vdash \sigma'\sigma : \Gamma''\theta &
            ( \Gamma''\theta'_i \vdashm \pi_i : P'\theta'_i )_i
            } \]
\end{proof}

The most interesting reduction rules are those for
the least and greatest fixed-point operators. In the former case, the
rule must apply to a proof of the form
\[\infer{\Gamma \vdash \delta_\mu(\mu(B,\t,\pi),\x.\alpha.\pi'):S\app \t}{
      \infer{\Gamma \vdash \mu(B,\t,\pi):\mu\app B\app \t}
            {\Gamma \vdash \pi:B\app(\mu\app B)\app \t} &
      \Gamma, \alpha : B\app S\app \x \vdash \pi' : S\app \x}
\]
This redex can be eliminated by generating a proof of $\Gamma \vdash
S\app \t$ directly from the derivation of $\Gamma \vdash \pi: B\app
(\mu\app B)\app \t$: doing this effectively means that we move the
redex (cut) deeper into the iteration that introduces the least fixed
point. To realize this transformation, we proceed as follows:
\begin{itemize} 
\item Using the derivation $\pi'$,
  we can get a proof of $S\app \t$ from $B\app S\app
  \t$. Thus, the task reduces to generating a proof of $B\app S\app
  \t$ from $B\app (\mu\app B)\app \t$.

\item Using again $\pi'$,
  we get a derivation for $\Gamma, \beta:\mu\app B\app \x \vdash
  \delta_\mu(\beta,\x.\alpha.\pi'):S\app \x$. If we can show how to
  ``lift'' this derivation over the operator $\lambda p. (B \app p\app
  \t)$, we obtain the needed derivation of $B\app S\app
  \t$ from $\pi : B\app (\mu\app B)\app \t$.
\end{itemize}
For the latter step, we use the notion of
{\it functoriality}~\cite{matthes98csl}. For any
monotonic operator $B$, we define the functor $F_B$ for which
the following typing rule is admissible:
\[ \infer{\Gamma \vdashm \F{B}(\x.\alpha.\pi) : (B\app P) \supset (B\app P')}
         {\Gamma, \alpha : P\app \x \vdash \pi : P'\app \x} \]

\begin{definition}[Functoriality, $\F{B}(\pi)$]
Let $B$ be an operator of type $(\vec{\gamma}\ra o)\ra o$,
and $\pi$ be a proof such that
$\alpha:P\app \x \vdash \pi:P'\app \x$.
We define $\F{B}^{+}(\x.\alpha.\pi)$ of type $B\app P \supset B\app P'$
for a monotonic $B$
and $\F{B}^{-}(\x.\alpha.\pi)$ of type $B\app P' \supset B\app P$
for an antimonotonic $B$
by induction on the maximum depth of an
occurrence of $p$ in $B \ap p$ through the rules in
Figure~\ref{fig:functoriality}. In these rules, $*$ denotes any
polarity ($+$ or $-$) and $-{*}$ denotes the complementary one.
We write $\F{B}^{+}(\x.\alpha.\pi)$ more simply as $\F{B}(\x.\alpha.\pi)$.
\end{definition}

\begin{figure*}[!t]
\[ \begin{array}{rcl}
  \multicolumn{3}{c}{
  \F{\lambda p. p\t}^{+}(\x.\alpha.\pi) = \lambda\alpha.\pi\subst{\t/\x} 
  \qquad\qquad
  \F{\lambda p. Q}^{*}(\x.\alpha.\pi) = \lambda\beta.\beta
    \mbox{ if $p$ does not occur in $Q$}
  } \quad \\
  \F{\lambda p. (B_1 \ap p) \wedge (B_2 \ap p)}^{*}(\x.\alpha.\pi) &=& 
    \lambda\beta.\pair{
                 \F{B_1}^{*}(\x.\alpha.\pi)\app (\fst{\beta}),
                 \F{B_2}^{*}(\x.\alpha.\pi)\app (\snd{\beta})} \\
  \F{\lambda p. (B_1 \ap p) \vee (B_2 \ap p)}^{*}(\x.\alpha.\pi) &=& 
    \lambda\beta.\elimv{\beta}{\gamma.\vin_1(\F{B_1}^{*}(\x.\alpha.\pi)\app
      \gamma)}{
                               \gamma.\vin_2(\F{B_2}^{*}(\x.\alpha.\pi)\app
                               \gamma)}
    \\ 
  \F{\lambda p. (B_1 \ap p) \supset (B_2 \ap p)}^{*}(\x.\alpha.\pi) &=&
    \lambda\beta. \lambda \gamma.
       \F{B_2}^{*}(\x.\alpha.\pi) \app (\beta \app (\F{B_1}^{-{*}}(\x.\alpha.\pi)\app \gamma)) \\
  \F{\lambda p. \forall x. (B\ap p \ap x)}^{*}(\x.\alpha.\pi) &=&
    \lambda\beta. \lambda x. \F{\lambda p. B\ap p\ap x}^{*}(\x.\alpha.\pi) \app (\beta \app x) \\
  \F{\lambda p. \exists x. (B \ap p \ap x)}^{*}(\x.\alpha.\pi) &=&
    \lambda\beta. \elimex{\beta}{
                     x.\gamma. \pair{x, \F{\lambda p. B\ap p\ap x}^{*}(\x.\alpha.\pi) \app \gamma}} \\
  \F{\lambda p. \mu\ap (B\app p)\ap \t}^{*}(\x.\alpha.\pi) &=&
    \lambda\beta. \delta_\mu(\beta, \x.\gamma. \mu (B\app P',\x,
      \F{\lambda p. B\app p\app (\mu\ap (B\app P'))\ap \x}^{*}(\x.\alpha.\pi)\app \gamma))
\\
  \F{\lambda p. \nu\ap (B\app p)\ap \t}^{*}(\x.\alpha.\pi) &=&
    \lambda\beta. \nu(\beta, \x.\gamma.
      \F{(\lambda p. B\app p\app (\nu\ap (B\app P))\ap \x)}^{*}(\x.\alpha.\pi)\app \delta_\nu(B\app P,\x,\gamma))
\end{array} \]
\caption{Definition of functoriality}
\label{fig:functoriality}
\end{figure*}

Checking the admissibility of the typing rule pertaining to $F_B$ is mostly
routine. We illustrate how this is to be done by considering
the least fixed point case in Figure~\ref{fig:Fmu};
the greatest fixed point case is shown in Figure~\ref{fig:Fnu} in the 
appendices.

\begin{figure*}[!t]
\centering
\[ \infer{\Gamma \vdash F^{+}_{\lambda p. \mu\ap (B\app p)\ap \t}(\x.\alpha.\pi) :
          \mu\ap (B\app P)\ap \t \supset \mu\ap (B\app P')\ap \t}{
   \infer{\Gamma, \beta : \mu\ap (B\app P)\ap \t \vdash \delta_\mu(\beta,\ldots) : \mu\ap (B\app P')\ap \t}{
     \infer{\Gamma, \beta : \mu\ap (B\app P)\ap \t \vdash \beta : \mu\ap (B\app P)\ap \t}{} &
     \infer{\Gamma, \beta : \mu\ap (B\app P)\ap \t, \gamma : B\app P\app (\mu\ap (B\app P'))\app \x
               \vdash \mu(B\app P',\x,\ldots) : \mu\ap (B\app P')\ap \x \vphantom{\t}}{
            \Gamma, \beta : \mu\ap (B\app P)\ap \t, \gamma : B\app P\app (\mu\ap (B\app P'))\app \x \vdash
            F_{\lambda p. B\app p\app (\mu\ap (B\app P'))\app \x}(\x.\alpha.\pi)\app\gamma :
            B\app P'\app (\mu\ap (B\app P'))\app \x}}} \]
\caption{Typing functoriality for least fixed-points}
\label{fig:Fmu}
\end{figure*}

\begin{figure*}[!t]
\[ \renewcommand{\arraystretch}{1.2} \begin{array}{rcl}
 \multicolumn{3}{c}{
  (\lambda \alpha.\pi) \app \pi' \;\ra\; \pi[\pi'/\alpha]
  \qquad
  \proj_i(\pair{\pi_1,\pi_2}) \;\ra\; \pi_i
  \qquad
  \elimv{\vin_i(\pi)}{\alpha.\pi_1}{\alpha.\pi_2} \;\ra\; \pi_i[\pi/\alpha]
 }
 \\
 \multicolumn{3}{c}{
  (\lambda x.\pi) \app t \;\ra\; \pi[t/x]
  \qquad
  \elimex{\pair{t,\pi}}{x.\alpha.\pi'} \;\ra\; \pi'\subst{t/x}\subst{\pi/\alpha} 
 }
 \\
  \delta_\mu(\mu(B,\t,\pi),\x.\alpha.\pi') &\ra&
     \pi'[\t/\x]\subst{\big(
       \F{\lambda p. B\ap p\ap \t}(
           \x.\beta.\delta_\mu(\beta,\x.\alpha.\pi'))\app \pi \big)/\alpha}
\\
  \delta_\nu(B,\t,\nu(\pi,\x.\alpha.\pi')) &\ra&
    \F{\lambda p. B\ap p\ap \t}(\x.\beta.\nu(\beta,\x.\alpha.\pi'))
    \app(\pi'\subst{\t/\x}\subst{\pi/\alpha}) \\
  \elimeq{\Gamma,\theta,\sigma,u,v,P,\refl{w}}{(\theta'_i.\pi_i)_i}
     &\ra& \pi_i\theta''\sigma \mbox{ where } \theta = \theta'_i\theta''
\end{array} \]
\caption{Reduction rules for \muNJ\ proof terms}
\label{fig:reduction}
\end{figure*}

The full collection of reduction rules is presented in
Figure~\ref{fig:reduction}. 
Note that the reduction rule for equality is not deterministic as
stated: determinism can be forced if needed by suitable assumptions
on \textcsus\ or by forcing a particular choice of $\theta'_i$ and
$\theta''$ in case of multiple possibilities. 

\begin{theorem}[Subject reduction]
If $\Gamma\vdashm\pi:P$ and $\pi\ra\pi'$ then $\Gamma\vdashm\pi':P$.
\end{theorem}

\begin{proof}
This follows from the above substitution lemmas. For example, consider 
the equality case. If $u\theta \equiv v\theta$ then
$\elimeq{\Gamma',\theta,\sigma,u,v,P,\refl{u\theta}}{(\theta'_i.\pi'_i)_i}
\ra \pi'_i\theta''\sigma$ where $\theta = \theta'_i\theta''$.
We have a derivation of $\Gamma'\theta'_i \vdashm \pi'_i : P\theta'_i$.
Hence, by applying $\theta''$ and using
Lemma~\ref{lem:termsubstred}, $\Gamma'\theta \vdashm \pi'_i\theta'' : P\theta$
must have a derivation.
Finally, since $\Gamma \vdashm \sigma : \Gamma' \theta$ has a
derivation, by Lemma~\ref{lem:proofsubstred} there must be one for $\Gamma
\vdashm \pi'_i\theta''\sigma : P\theta$.
\end{proof}

\begin{proposition} \label{prop:redsubst}
For any proof terms $\pi$, $\pi'$ and $\rho$ and any term $t$,
$\pi \ra \pi'$ implies $\pi[\rho/\alpha] \ra \pi'[\rho/\alpha]$
and $\pi[t/x] \ra \pi'[t/x]$.
\end{proposition}

\begin{proof}
Both implications are easily checked.
\end{proof}

A proof term is {\em normal} if it contains no redexes and it is
strongly normalizable if every reduction sequence starting from it
terminates in a normal proof term. The set of strongly normalizable
proof terms is denoted by $\SN$. The normalizability of proof terms
can be coupled with the following observation to show the
(conditional) consistency of the logic. 

\begin{lemma} \label{lem:bot}
If ${\equiv}$ is defined by a confluent rewrite system 
that rewrites terms to terms and atomic propositions to propositions,
then $\vdashm \pi : \bot$ is not derivable for any normal $\pi$. 
\end{lemma}

\begin{proof}
This standard observation is not affected by the rewriting layer,
since $\bot$ cannot be equated with another logical connective under the
assumptions on $\equiv$, and it is not affected either by our new constructs,
for which progress is ensured: eliminations followed by introductions can
always be reduced. More details may be found in Appendix~I.
\end{proof}

%% file: norm.tex
\section{Strong Normalizability} \label{sec:norm}

In a fashion similar to \cite{dowek03jsl}, we now establish strong
normalizability for proof reductions when the congruence
relation satisfies certain general conditions.
The proof is based on the framework of reducibility
candidates, and borrows elements from earlier work in linear logic
\cite{baelde12tocl} regarding fixed-points.

\begin{definition}
A proof term is \emph{neutral} iff it is not an introduction,
\ie\ it is a variable or an elimination construct.
\end{definition}

\begin{definition}
A set $R$ of proof terms is a \emph{reducibility candidate} if
(1) $R\subseteq\SN$;
(2) 
  $\pi\in R$ and $\pi\ra\pi'$ implies $\pi'\in R$;
and (3) if $\pi$ is neutral and all of its one-step
  reducts are in $R$, then $\pi\in R$.
We denote by $\C$ the set of all reducibility candidates.
\end{definition}

Conditions (2,3) are positive and compatible with (1)
so that for any subset $S$ of $\SN$ there is a least candidate
containing $S$. 
We refer to the operation that yields this set as \emph{saturation}.
Reducibility candidates, equipped with inclusion, form a complete lattice:
the intersection of a family of candidates gives their infimum
and the saturated union gives their supremum.
Having a complete lattice, we can define least and
greatest fixed points of monotonic operators.
The ordering and the observations about it lift pointwise for
functions from terms to candidates, which we call \emph{predicate
  candidates}. We use $\X$ and $\Y$ ambiguously to denote candidates
and predicate candidates. 

\begin{definition}
A \emph{pre-model} $\M$ is an assignment of a function $\hat{a}$ from 
$|\gamma_1|\times\ldots\times|\gamma_n|$ to $\C$ to each predicate
constant $a$ of type $\gamma_1\ra\ldots\gamma_n\ra o$. Here, 
$|\gamma|$ denotes the set of (potentially open) terms of type $\gamma$.
\end{definition}


\begin{definition} \label{def:interp}
Let $\M$ be a pre-model, let $P$ be a formula and let $\E$ be a
context assigning predicate candidates of the right types to at
least the free predicate variables in $P$. 
We define the candidate $\interp{P}$,
called the \emph{interpretation of} $P$, by recursion on the structure
of $P$ as shown in Figure~\ref{fig:interpretation}. 
\end{definition}

\begin{figure*}[!t]
\[\begin{array}{rcl}
\multicolumn{3}{c}{
  \interp{\bot} = \interp{\top} = \interp{u=v} = \SN
\qquad
  \interp{p \ t_1 \ldots t_n} = \E(p)(t_1,\ldots,t_n)
\qquad
  \interp{a \ t_1 \ldots t_n} = \hat{a}(t_1,\ldots,t_n)
}
\\
  \interp{P \supset Q} &=& \set{\pi\in\SN}{
     \pi \ra^* \lambda\alpha.\pi_1
     \mbox{ implies }
     \pi_1[\pi'/\alpha]\in\interp{Q}
     \mbox{ for any }
     \pi'\in\interp{P}
   }
\\
  \interp{P \wedge Q} &=& \set{\pi\in\SN}{
     \pi \ra^* \pair{\pi_1,\pi_2} \mbox{ implies }
     \pi_1\in\interp{P}
     \mbox{ and }
     \pi_2\in\interp{Q}
   }
\\
  \interp{P_1 \vee P_2} &=& \set{\pi\in\SN}{
     \pi \ra^* \vin_i(\pi') \mbox{ implies }
     \pi'\in\interp{P_i}
   }
\\
  \interp{\forall x.~ P} &=& \set{\pi\in\SN}{
    \pi \ra^* \lambda x.\pi'
    \mbox{ implies }
    \pi'[t/x]\in\interp{P[t/x]}
    \mbox{ for any }
    t}
\\
  \interp{\exists x.~ P} &=& \set{\pi\in\SN}{
    \pi \ra^* \pair{t,\pi'}
    \mbox{ implies }
    \pi'[t/x]\in\interp{P[t/x]}
  }
\\
  \interp{\mu B \t} &=& \lfp{\phi}(\t) \mbox{ where } 
  \phi(\X) = \t'\mapsto \set{\pi\in\SN}{
       \pi \ra^* \mu (B,\t', \pi') \text{ implies }
       \pi' \in \interp[\E+\pair{p,\X}]{B p \t'}
    }
\\
  \interp{\nu B \t} &=& \gfp{\phi}(\t) \mbox{ where } 
  \phi(\X) = \t'\mapsto \set{\pi}{
     \delta_\nu(B,\vec{t'},\pi)\in\interp[\E+\pair{p,\X}]{B p \t'}
  }
\ignore{
\\
  \interp{\nu B \t} &=& \gfp{\phi_{\nu B}} \\
  \phi_{\nu B}(\X) &=&
    \{\; \pi\in\SN \; : \;
       \pi \ra^* \nu(\pi',\x.\alpha.\pi'') \text { implies } \\
  & &
    \quad\quad
       (\lambda p.~ B p \t')(
          \x.\alpha.\nu(\alpha, \x.\alpha.\pi''))
        (\pi''\subst{\t'/\x}\subst{\pi'/\alpha})
       \in \interp[\E+\pair{p,\X}]{B p \t'}
    \; \}
}
\end{array}\]
\caption{Interpretation of formulas as candidates}
\label{fig:interpretation}
\end{figure*}

To justify this definition, we show simultaneously by an induction on
$P$ that $\interp{P}$ is a
candidate and that it is monotonic (resp. anti-monotonic) 
in $\E(p)$ for any variable $p$ that only occurs positively
(resp. negatively) in~$P$; the latter two facts ensure
that the fixed points assumed in the definition actually exist,
anti-monotonicity being needed because of the covariance in
implication formulas.
Preservation of (anti)monotonicity and satisfaction of the conditions
for reducibility candidates are readily verified in all but the fixed point cases.
For the least fixed point case, $\interp{\mu\ap B\ap \t}$ is easily
seen to be a candidate provided it is well-defined, \ie\ if
$\lfp{\phi}$ exists for $\phi$ as in the definition. 
But this must be so: the induction hypothesis applied to $B\ap p\ap
\t'$ ensures that $\phi$ is a monotonic mapping,
hence it has a least fixed point in the lattice of
predicate candidates. 
For monotonicity, consider $\E$ and $\E'$ differing only on a variable
$p$ 
that occurs only positively in $\mu \ap B \ap \t$,
with $\E(p)\subseteq \E'(p)$. Let $\interp[\E']{\mu\ap B\ap \t} =
\lfp{\phi'} \t$. Unfolding and using the induction hypothesis, we have
$\phi(\X) \subseteq \phi'(\X)$ for any candidate $\X$,
and in particular $\phi(\interp[\E']{\mu\ap B\ap \t}) \subseteq
\phi'(\interp[\E']{\mu\ap B\ap \t}) = \interp[\E']{\mu\ap B\ap \t}$.
The least fixed point being contained in all prefixed points,
we obtain the expected result:
$\interp[\E]{\mu\ap B\ap \t} = \lfp{\phi} \subseteq \interp[\E']{\mu\ap B\ap \t}$.
Antimonotonicity is established in a symmetric fashion. The treatment
of the greatest fixed point case is similar. 


\begin{notation}
If $P$ is a predicate of type $\vec{\gamma}\ra o$,
$\interp{P}$ denotes the mapping $\vec{t} \mapsto \interp{P\ap \t}$.
If $B$ is of type $(\vec{\gamma}\ra o)\ra o$, $\interp{B}$
denotes the mapping $\X \mapsto \interp[\E + \pair{p,\X}]{B\ap p}$ and
if $B$ is a predicate operator of type $(\vec{\gamma}\ra o)\ra
\vec{\gamma} \ra o$, $\interp{B}$ denotes the mapping $\X \mapsto \t \mapsto
\interp[\E + \pair{p,\X}]{B\ap p\ap \t}$.
For conciseness we write directly
$\interp{B\ap \X\ap \t}$ for 
$\interp{\lambda p.~ B\ap p\ap \t}\ap \X$ or, equivalently,
$\interp{B}\ap \X\ap \t$.
\end{notation}

\begin{lemma} \label{lem:subst}
Interpretation commutes with second-order substitution:
$\interp{B\subst{P/p}} = \interp[\E+\pair{p, \interp{P}}]{B}$.
\end{lemma}

\begin{proof}
Straightforward, by induction on $B$.
\end{proof}

We naturally extend the interpretation to typing contexts:
if $\Gamma = (\alpha_1:P_1, \ldots, \alpha_n:P_n)$,
$\interp{\Gamma} =
(\alpha_1:\interp{P_1}, \ldots, \alpha_n:\interp{P_n})$.
We also write $\sigma \in \interp{\Gamma}$ when
$\sigma$ is of the form
$[\pi_1/\alpha_1,\ldots,\pi_n/\alpha_n]$
with $\pi_i\in\interp{P_i}$ for all $i$.

\begin{definition}
If $\pi$ is a proof term with free variables $\alpha_1, \ldots, \alpha_n$
and $\Y, \X_1, \ldots, \X_n$ are reducibility candidates,
we say that $\pi$ is
\emph{$(\alpha_1:\X_1, \ldots, \alpha_n:\X_n \vdash \Y)$-reducible}
if $\pi\subst{\pi'_i/\alpha_i}_i\in\Y$ for any $(\pi'_i)_i \in (\X_i)_i$.
When it is not ambiguous, we may omit the variables and simply
say that $\pi$ is $(\X_1,\ldots,\X_n\vdash \Y)$-reducible.
\end{definition}

\begin{definition}
A pre-model $\M$ is a \emph{pre-model of $\equiv$}
iff it accords the same interpretation to formulas that are congruent.
\end{definition}

In the rest of this section, we assume that $\M$ is a pre-model of the
congruence, and we show that if $\Gamma \vdash \pi : P$ has a proof
then $\pi$ is $(\interp{\Gamma}\vdash\interp{P})$-reducible.
In order to do so, we prove \emph{adequacy} lemmas which show that
each typing rule can be simulated in the interpretation.

\begin{lemma} \label{lem:adequacies}
The following properties hold for any context $\E$.
\begin{itemize}
\item[(${\supset}$)]
  \begin{itemize}
  \item If $\pi$ is $(\alpha:\interp{P} \vdash \interp{Q})$-reducible, \\ then
    $\lambda \alpha.\pi \in \interp{P \supset Q}$.
  \item If $\pi\in\interp{P\supset Q}$ and $\pi'\in\interp{P}$, then
    $\pi\app\pi'\in\interp{Q}$.
  \end{itemize}
\item[(${\wedge}$)]
  \begin{itemize}
  \item If $\pi_1\in\interp{P_1}$ and $\pi_2\in\interp{P_2}$, then
    $\pair{\pi_1,\pi_2}\in\interp{P_1\wedge P_2}$.
  \item If $\pi\in\interp{P_1\wedge P_2}$, \\ then
    $\fst{\pi}\in\interp{P_1}$ and $\snd{\pi}\in\interp{P_2}$.
  \end{itemize}
\item[(${\vee}$)]
  \begin{itemize}
  \item If $\pi\in\interp{P_i}$ for $i\in\{1,2\}$, then
    $\vin_i(\pi)\in\interp{P_1\vee P_2}$.
  \item If $\pi\in\interp{P_1\vee P_2}$ and each $\pi_i$
    is $(\alpha:\interp{P_i}\vdash \interp{Q})$-reducible,
    then $\elimv{\pi}{\alpha.\pi_1}{\alpha.\pi_2} \in \interp{Q}$.
  \end{itemize}
\item[($\top$)]
  \begin{itemize}
  \item The proof $\unit$ belongs to $\interp{\top}$.
  \end{itemize}
\item[($\bot$)]
  \begin{itemize}
  \item If $\pi\in\interp{\bot}$, then $\elimbot{\pi}\in\interp{P}$
     for any $P$.
  \end{itemize}
\item[(${\forall}$)]
  \begin{itemize}
  \item If $\pi[t/x] \in \interp{P[t/x]}$ for any $t$, then
    $\lambda x.\pi \in \interp{\forall x.~ P}$.
  \item If $\pi\in\interp{\forall x.~ P}$, then $\pi\app t \in \interp{P[t/x]}$.
  \end{itemize}
\item[(${\exists}$)]
  \begin{itemize}
  \item If $\pi \in \interp{P[t/x]}$,
    then $\pair{t,\pi}\in\interp{\exists x.~ P}$.
  \item If $\pi \in \interp{\exists x.~ P}$ and $\pi'\subst{t/x}$ is 
    $(\alpha:\interp{P\subst{t/x}}\vdash \interp{Q})$-reducible for any $t$, then
    $\elimex{\pi}{x.\alpha.\pi'} \in \interp{Q}$.
  \end{itemize}
\item[(${=}$)]
  \begin{itemize}
  \item 
        $\refl{t} \in \interp{t=t}$.
  \item If $\pi \in \interp{t\theta = t'\theta}$, $\sigma \in 
    \interp{\Gamma\theta}$ and
    $\pi'_i\theta'$ is
    $(\interp{\Gamma\theta_i\theta'}\vdash\interp{P\theta_i\theta'})$-reducible
    for any $i$ and $\theta'$,
    then $\elimeq{\Gamma,\theta,\sigma,t,t',P,\pi}{(\theta_i.\pi_i)_i} \in \interp{P\theta}$.
  \end{itemize}
\item[($\mu$)]
  \begin{itemize}
  \item If $\pi\in\interp{B(\mu B)\t}$,
     then $\mu(B,\t,\pi)\in\interp{\mu B \t}$.
  \end{itemize}
\item[($\nu$)]
  \begin{itemize}
  \item If $\pi\in\interp{\nu B \t}$,
    then $\delta_\nu(B,\t,\pi)\in \interp{B (\nu B) \t}$.
  \end{itemize}
\end{itemize}
\end{lemma}

\begin{proof}
These observations are proved easily using standard proof techniques on 
candidates. We illustrate only a few cases here;
more details may be found in Appendix~I.
For the case of least fixed point introductions, we have
$\interp{\mu B\t} = \lfp{\phi}(\t) = \phi(\interp{\mu B})(\t)$
by Definition~\ref{def:interp}, and thus
$\interp{\mu B \t} =
 \set{\rho\in\SN}{\rho \ra^* \mu(B,\t,\pi') \mbox{ implies } \pi'\in
 \interp{B(\mu B)\t}}$ by Lemma~\ref{lem:subst},
from which it is easy to conclude.
Similarly, we observe that
$\interp{\nu B \t} = \set{\pi}{\delta_\nu(B,\t,\pi)\in\interp{B(\nu B)\t}}$
from which the greatest fixed point elimination case follows immediately.
Finally, the equality elimination case is proved by induction on
the strong normalizability of the subderivations $\pi$, $\sigma$ and $\pi_i$.
In order to show that a neutral term belongs to a candidate, it suffices to
consider all its one-step reducts. Reductions occurring inside subterms are 
handled by induction hypothesis. We may also have a toplevel redex when
$t\theta \equiv t'\theta$ and $\pi = \refl{t\theta}$, reducing to
$\pi_i\theta'\sigma$ where $\theta'$ is such that
$\theta_i\theta' \equiv \theta$.
By hypothesis, $\pi_i\theta'$ is
$(\interp{\Gamma\theta_i\theta'} \vdash \interp{P\theta_i\theta'})$-reducible
and $\sigma \in \interp{\Gamma\theta} = \interp{\Gamma\theta_i\theta'}$,
and thus we have $\pi_i\theta'\sigma\in\interp{P\theta}$ as expected.
\end{proof}

Although adequacy is easily proved for our new equality formulation,
a few important observations should be made here.
First, the proof crucially relies on the fact that we are considering only 
syntactic pre-models, and not the general notion of pre-model of Dowek and 
Werner where terms may be interpreted in arbitrary structures.
This requirement makes sense conceptually, since closed-world
equality internalizes the fact that equality can only hold when the
congruence allows it, and is thus incompatible with further equalities
that could hold in non-trivial semantic interpretations.
Second, the suspension of proof-level substitutions in equality elimination
goes hand in hand with the independence of interpretations for
different predicate instances, which in turn is necessary to interpret
recursive definitions.
Indeed, when applying a proof-level substitution $\sigma\in\interp{\Gamma}$
on an eager equality elimination, we are forced to apply the \textcsu\ 
substitutions on $\sigma$, and we need $\sigma \in \interp{\Gamma\theta_i}$
which essentially forces us to have a term-independent
interpretation~\cite{baelde12tocl}.


\renewcommand{\interp}[2][]{|#2|^{#1}}

We now address the adequacy of functoriality, induction and coinduction.

\begin{lemma} \label{lem:adeqfun}
Let $\pi$ be a proof, and let $\X$ and $\X'$ be predicate candidates
such that
$\pi\subst{\t/\x}$ is $(\alpha:\X\t \vdash \X'\t)$-reducible
for any $\t$.
If $B$ is a monotonic (resp.~antimonotonic) operator,
then $\F{B}^+(\x.\alpha.\pi) \in \interp{B\X \supset B\X'}$
(resp.
$\F{B}^-(\x.\alpha.\pi) \in \interp{B\X'\supset B\X}$).
\end{lemma}

\begin{lemma} \label{lem:adeqmu}
Let $\pi$ be a proof and $\X$ a predicate candidate.
If
$\pi\subst{\t/\x}$ is $(\alpha:\interp{B}\X\t\vdash \X\t)$-reducible
for any $\t$,
then
$\delta_\mu(\beta,\x.\alpha.\pi)$ is
$(\beta:\interp{\mu B \t'} \vdash \X\t')$-reducible
for any $\t'$.
\end{lemma}

\begin{lemma} \label{lem:adeqnu}
Let $\pi$ be a proof and $\X$ a predicate candidate.
If
$\pi\subst{\t/\x}$ is $(\alpha:\X\t\vdash \interp{B}\X\t)$-reducible
for any $\t$,
then
$\nu(\beta,\x.\alpha.\pi)$ is
$(\beta:\X\t' \vdash \interp{\nu B \t'})$-reducible
for any $\t'$.
\end{lemma}

\ignore{
\begin{lemma}
$\beta.\nu(B,\t,\beta,\x.\alpha.\pi)$ is
$(\X\vdash \interp{\nu B \t})$-reducible when
$\pi$ is $(\X\vdash\interp{B}\X)$-reducible.
\end{lemma}
}

\begin{proof}
These lemmas must be proved simultaneously,
in a generalized form that is detailed in Appendix~I.
There is no essential difficulty in proving the functoriality
lemma, using previously proved adequacy properties as well as
the other two lemmas for the fixed point cases.
The next two lemmas are the interesting ones, since they involve
using the properties of the fixed point interpretations to 
justify the (co)induction rules.
In the case of induction, we need to establish that
  $\delta_\mu(\rho,\x.\alpha.\pi) \in \X\t$ when
  $\rho\in\interp{\mu B \t}$.
In order to do this, it suffices to show that
$ \Y := \t \mapsto \set{\rho}{
         \delta_\mu(\rho,\x.\alpha.\pi) \in \X\t } $
is included in $\interp{\mu B}$.
This follows from the fact that
$\Y$ is a pre-fixed point of the operator $\phi$ such that
$\interp{\mu B} = \lfp{\phi}$, which can be proved
easily using the adequacy property for functoriality.
We proceed similarly for the coinduction rule,
showing that 
  \[ \begin{array}{l} \Y := \t \mapsto \set{\pi\in\SN}{
     \pi \ra^* \nu(\rho,\x.\alpha.\pi)
     \text { implies }
      \rho\in \X\t
      \mbox{ and } \\ \quad
      \pi[\t'/\x] \mbox{ is ($\alpha:\X\t'\vdash \interp{B}\X\t'$)-reducible
      for any $\t'$}}
  \end{array} \]
is a post-fixed point of the operator $\phi$ such
that $\interp{\nu B} = \gfp{\phi}$.
In both cases,
note that the candidate $\Y$ is \emph{a priori} not the interpretation
of any predicate; this is where we use the power of reducibility candidates.
\end{proof}


\begin{theorem}[Adequacy] \label{th:adequacy}
Let $\equiv$ be a congruence, $\M$ be a pre-model of $\equiv$ and
$\Gamma\vdashm\pi:P$ be a derivable judgment.
Then $\pi\sigma\in\interp{P}$ for any
substitution $\sigma \in \interp{\Gamma}$.
\end{theorem}

\begin{proof}
By induction on the height of $\pi$,
using the previous adequacy properties.
\end{proof}

The usual corollaries hold.
Since variables belong to any candidate by condition (3),
we can take $\sigma$ to be the identity substitution,
and obtain that any well-typed proof is strongly normalizable.
Together with Lemma~\ref{lem:bot}, this means that our logic is
consistent.
Note that the suspended computations in the (co)induction
and equality elimination rules do not affect these corollaries,
because they can only occur in normal forms of specific types.
For instance, equality elimination cannot hide a non-terminating
computation if there is no equality assumption in the context.

%% file: application.tex
\section{Recursive Definitions}\label{sec:rec-defs}

We now identify a class of rewrite rules relative to which we can
always build a pre-model.
This class supports recursive definitions whose use 
we illustrate through a sound formalization of a Tait-style argument. 


\subsection{Recursive rewriting that admits a pre-model}

The essential idea behind recursive definitions is that they are formed
gradually, following the inductive structure of one of their arguments,
or more generally a well-founded order on arguments.
In order to reflect this idea into a pre-model construction, we need
to identify all the atom interpretations that could be involved in 
the interpretation of a given formula. This is the purpose of the next
definition. 

\begin{definition}
We say that $P$ \emph{may occur} in $Q$ when $P = P'\theta$,
$P'$ occurs in $Q$, and $\theta$ is a substitution for
variables quantified over in $Q$.
\end{definition}

For example, $(a\ap t)$ may occur in $(a'\ap x \wedge \exists y.~ a\ap y)$
for any $t$. 

\begin{theorem} \label{th:rec}
Let ${\equiv}$ be a congruence defined by a rewrite system
rewriting terms to terms and atomic propositions to propositions,
and let $\M$ be a pre-model of ${\equiv}$.
Consider the addition of new predicate symbols $a_1, \ldots, a_n$ in
the language, together with the extension of the congruence resulting
from the addition of rewrite rules of the form $a_i \vec{t} \rew B$.
There is a pre-model of the extended congruence in
the extended language, provided that the following conditions hold.
\begin{itemize}
\item[(1)]
  If $(a_i\vec{t})\theta \equiv (a_i\vec{t'})\theta'$,
  $a_i\vec{t}\rew B$ and $a_i\vec{t'}\rew B'$,
  then
$B\theta \equiv B'\theta'$.
\item[(2)] There exists a well-founded order $\prec$ such that
  $a_j\vec{t'} \prec (a_i\vec{t})\theta$ whenever
  $a_i\vec{t} \rew B$ and $a_j\vec{t'}$ may occur in $B'\theta$.
\end{itemize}
\end{theorem}

Note that condition (1) is not obviously satisfied, even when there
is a single rule per atom.
Consider, for example, $a \ap (0 \times x) \rew a' \ap x$ in a setting
where $0 \times x \equiv 0$: our condition requires that $a' \ap x \equiv a' \ap y$
for any $x$ and $y$, which is \emph{a priori} not guaranteed.
Condition (2) restricts the use of quantifiers but still allows useful 
constructions. Consider for example the Ackermann relation, built using a 
double induction on its first two parameters:
$ack~0~x~(s~x) \rew \top$,
$ack~(s~x)~0~y \rew ack~x~(s~0)~y$ and
$ack~(s~x)~(s~y)~z \rew \exists r.~ ack~(s~x)~y~r \wedge ack~x~r~z$.
The third rule requires that
$ack~x~r~z \prec ack~(s~x)~(s~y)~z$ for any $x$, $y$, $z$ and $r$,
which is indeed satisfied with a lexicographic ordering.

\begin{proof}
We only present the main idea here;
a detailed proof may be found in Appendix~I.
We first build pre-models $\M^{a_i\vec{t}}$ that
are compatible with instances $a_j\vec{t'}\rew B$
of the new rewrite rules for $a_j\vec{t'}\preceq a_i\vec{t}$.
This is done gradually following the order $\prec$,
using a well-founded induction on $a_i\vec{t}$.
We build $\M^{a_i\vec{t}}$ by aggregating smaller pre-models
$\M^{a_j\vec{t'}}$ for $a_j\vec{t'}\prec a_i\vec{t}$,
and adding the interpretation $\hat{a_i}\vec{t}$.
To define it, we consider rule instances of the form $a_i\vec{t} \rew B$.
If there is none we use a dummy interpretation:
$\hat{a_i}\vec{t} = \SN$. Otherwise, condition (1) imposes
that there is essentially a single possible such rewrite
modulo the congruence, so it suffices to choose $\interp{B}$
as the interpretation $\hat{a_i}\vec{t}$ to satisfy
the new rewrite rules.
Finally, we aggregate interpretations from
all the pre-models $\M^{a_i\vec{t}}$ to obtain a pre-model
of the full extended congruence.
\end{proof}

This result can be used to obtain pre-models for complex definition schemes,
such as ones that iterate and interleave groups of fixed-point and recursive 
definitions. Consider, for example,
$a~(s~n) \rew a~n \supset a~(s~n)$.
While this rewrite rule does not directly satisfy the conditions of
Theorem~\ref{th:rec}, 
it can be rewritten into the form $a~(s~n) \rew \mu Q.~ a~n \supset
Q$, which does satisfy these conditions. 

\subsection{An application of recursive definitions}

Our example application is the formalization of the Tait-style
argument of strong normalizability for the simply typed
$\lambda$-calculus. 
We assume term-level sorts $tm$ and $ty$ corresponding to representations 
of $\lambda$-terms and simple types, and symbols $\iota : ty$, $arrow : ty\ra ty\ra 
ty$, $app : tm \ra tm \ra tm$ and $abs : (tm \ra tm) \ra tm$.
We identify well-formed types through an inductive predicate:
\[ isty \eqdef \mu (\lambda T \lambda t.~ t = \iota \vee
     \exists t' \exists t''.~ t = arrow~t'~t'' \wedge T~t' \wedge T~t'') \]
We assume a definition of term reduction and strong 
normalization, denoting the latter predicate by $sn$.
Finally, we define $red~m~t$, expressing that
$m$ is a reducible $\lambda$-term of type $t$,
by the following rewrite rules:
\begin{eqnarray*}
  red~m~\iota &\rew& sn~m \\
  red~m~(arrow~t~t') &\rew&
    \forall n.~ red~n~t \supset red~(app~m~n)~t'
\end{eqnarray*}
This definition satisfies the conditions of Theorem~\ref{th:rec},
taking as $\prec$ the order induced by the subterm ordering
on the second argument of $red$.
We can thus safely use it.

With these definitions, our logic allows us to mirror very closely
the strong normalization proof presented in \cite{girard89book}.
For instance, consider proving that reducible terms are strongly normalizing:
\[ \forall m \forall t.~ isty~t \supset red~m~t \supset sn~m \]
The paper proof is by induction on types,
which corresponds in the formalization to an elimination on $isty~t$.
In the base case, we have to derive $red~m~\iota \supset sn~m$
which is simply an instance of $P \supset P$ modulo our congruence.
In the arrow case, we must prove $red~m~(arrow~t~t') \supset sn~m$.
The hypothesis $red~m~(arrow~t~t')$ is congruent to
$\forall n.~ red~n~t \supset red~(app~m~n)~t'$ and we can
show that variables are always reducible,\footnote{
  This actually has to be proved simultaneously with $red~m~t \supset sn~m$,
  but we ignore it for the simplicity of the presentation.
} which gives us
$red~(app~m~x)~t'$. From there, we obtain $sn~(app~m~x)$ by induction 
hypothesis, from which we can deduce $sn~m$ with a little more work.

The full formalization, which is too detailed to present here, is shown
in Appendix~II. It has been tested
using the proof assistant Abella \cite{gacek08ijcar}.
The logic that underlies Abella features fixed-point
definitions, closed-world equality and generic quantification.
The last notion is useful when dealing with binding structures,
and we have employed it in our formalization although it is not 
available yet in our logic. 
Abella does not actually support recursive
definitions. To get around this fact, we have entered the one we need
as an inductive definition, and ignored the warning provided 
about the non-monotonic clause while making sure to use an unfolding
of this inductive definition in the proof only when this is allowed
for recursive definitions. In the future, we plan to extend Abella to
support recursive definitions based on the theory developed in this
paper. This would mean allowing such definitions as a separate class,
building in a test that they satisfy the criterion described in
Theorem~\ref{th:rec} and properly restricting the use of these
definitions in proofs. Such an extension is obviously compatible with
all the current capabilities of Abella and would support additional
reasoning that is justifiably sound.

%% file: new.tex
\section{Related and Future Work} \label{sec:conclusion}

The logical system that we have developed is obviously
related to deduction modulo. In essence, it extends that system with
a simple yet powerful treatment of fixed-point definitions. The
additional power is obtained from 
two new features: fixed-point combinators and closed-world
equality.
If our focus is only on provability, the capabilities arising
from these features may perhaps be encoded
in deduction modulo. Dowek and Werner provide an encoding of
arithmetic in deduction modulo, and also show how to build pre-models for
some more general fixed-point constructs~\cite{dowek03jsl}.
Regarding equality, Allali~\cite{allali07types} has shown that a
more algorithmic version of equality may be defined through the congruence, 
which allows to simplify some equations by computing.
Thus, it simulates some aspects of
closed-world equality. 
However, the principle of substitutivity has to be recovered
through a complex encoding involving inductions on the term structures.
In any case, our concern here is not simply with provability;
in general, we do not follow the project of deduction modulo to have
a logic as basic as possible in which stronger systems are then encoded.
Rather, we seek to obtain meaningful proof structures,
whose study can reveal useful information. For instance,
in the context of proof-search, it has been shown that a direct treatment
of fixed-point definitions allows for stronger focused proof 
systems~\cite{baelde12tocl} which have served as a basis for several 
proof-search implementations~\cite{baelde07cade,baelde10ijcar}.
This goal also justifies why we do not simply use powerful systems
such as the Calculus of Inductive Constructions~\cite{paulin93tlca}
which obviously supports inductive as well as recursive definitions;
here again we highlight the simplicity of our (co)induction rules 
and of our rich equality elimination principle.

Our logic is also related to logics of fixed-point 
definitions~\cite{schroeder-Heister93lics,mcdowell00tcs,momigliano08corr}.
The system we have described represents an advance over these logics in that
it adds to them a rewriting capability. As we have seen, this
capability can be used to blend computation and deduction in natural ways
and add support for recursive definitions --- a similar support may also
be obtained in other ways~\cite{tiu12ijcar}.
Our work also makes important contributions to the understanding of 
closed-world equality.
We have shown that it is compatible with an equational theory on
terms.
We have, in addition, resolved some problematic issues related to this
notion that affect the stability of finite proofs under reduction.
This has allowed us to prove for the first time a strong
normalizability result for logics of fixed-point definitions.
Our calculus is, at this stage, missing a treatment of generic
quantification present in some of the alternative logics
\cite{gacek09phd,gacek11ic,miller05tocl}. We plan to include this
feature in the future, and do not foresee any difficulty in doing so since
it has typically been added in a modular fashion to such logics.
This addition would make our logic an excellent choice for formalizing
the meta-theory of  computational and logical systems. 

An important topic for further investigation of our system is proof
search. 
The distinction between computation and deduction is critical
for theorem proving with fixed point definitions.
For instance, in the Tac system~\cite{baelde10ijcar}, which is based on logics of 
definitions, automated (co)inductive theorem proving relies heavily on
ad-hoc annotations that identify computations. In that context, our treatment 
of recursive definitions seems like a good candidate more a more principled 
separation of computation and deduction.
Finally, now that we have refactored equality rules to simplify the proof 
normalization process, we should study their proof search behavior.
The new equality elimination rule seems difficult to analyze at first. However, 
we hope to gain some insights from studying its use in
settings where the old rule (which it subsumes) is practically
satisfactory, progressively moving to newer contexts where it offers
advantages. 
We note in this regard that the new complexity is in fact welcome: the
earlier infinitely branching treatments of closed-world equality had a
simple proof-search treatment in theory, but did not provide a useful
handle to study the practical difficulties
of automated theorem proving with complex equalities.

%% file: appendixfigs.tex
\begin{figure*}[!]
\centering
\[ \infer{\Gamma \vdash \F{\lambda p. \nu\ap (B\app p)\ap \t}^{+}(\x.\alpha.\pi) :
          \nu\ap (B\app P)\ap \t \supset \nu\ap (B\app P')\ap \t}{
   \infer{\Gamma, \beta : \nu\ap (B \app P)\ap \t
     \vdash \nu(\beta,\ldots) : \nu \ap (B \app P') \app \t}{
   \infer{\ldots 
          \vdash \beta : \nu\ap(B\app P)\ap\t}{} &
   \infer{\vphantom{\t}
      \ldots, \gamma : \nu\ap(B\app P)\ap\x \vdash 
      (\F{(\lambda p. B\app p\app (\nu\ap (B\app P))\ap \x)}^{*}(\x.\alpha.\pi))\app \delta_\nu(B\app P,\x,\gamma)
        : B \app P' \app (\nu\ap(B\app P)) \x}{
   \ldots \vdash \delta_\nu(B\app P,\x,\gamma) : B \app P \app (\nu\ap(B\app P)) \app \x &
   \ldots \vdash \F{(\lambda p. B\app p\app (\nu\ap (B\app P))\ap \x)}^{+}(\x.\alpha.\pi) :
      B \app P \app (\nu\ap(B\app P)) \x \supset B \app P' \app (\nu\ap(B\app P)) \app \x
   }}} \]
\caption{Typing functoriality for greatest fixed-points}
\label{fig:Fnu}
\end{figure*}

%% file: proofs.tex
\section*{Appendix I: Proofs of Lemmas and Theorems} 
\setcounter{subsection}{0}

\subsection{Proof of Lemma~\ref{lem:bot}}

We first observe that typed normal forms are characterized as usual:
no introduction term is ever found as the main parameter of an elimination.
This standard property is not affected by our new constructs. For example,
consider the case of equality: $\elimeq{\ldots,\refl{t}}{(\theta_i.\pi_i)_i}$
can always be reduced by definition of complete sets of unifiers.
The rest of the argument follows the usual lines:
the proof cannot end with an elimination, otherwise it would have to be
a chain of eliminations terminated with a proof variable, but the context
is empty;
it also cannot end with an introduction since there is no
introduction for $\bot$ and the congruence cannot equate it
with another connective.

\subsection{Proof of Lemma~\ref{lem:adequacies}}

All introduction rules are treated in a similar fashion:
\begin{itemize}

\item
  \emph{If $\pi$ is $(\alpha:\interp{P} \vdash \interp{Q})$-reducible, then
    $\lambda \alpha.\pi \in \interp{P \supset Q}$.}

  First, $\lambda\alpha.\pi$ is \SN, like all reducible proof-terms,
  because variables belong to all candidates, and candidates are sets
  of \SN\ proofs.
  Now, assuming $\lambda\alpha.\pi\ra^* \lambda\alpha.\pi'$,
  we seek to establish that
  $\pi'[\pi''/\alpha] \in \interp{Q}$ for any $\pi''\in\interp{P}$.
  By definition of reducibility, $\pi[\pi''/\alpha]$ belongs to
  $\interp{Q}$, and we conclude by stability of candidates under reduction
  since $\pi[\pi''/\alpha]\ra^* \pi'[\pi''/\alpha]$.
  
\item The cases for $\wedge$, $\vee$ and $\exists$ are proved similarly.
\item The cases for $\top$ and equality are trivial.

\item
  \emph{If $\pi[t/x] \in \interp{P[t/x]}$ for any $t$, then
    $\lambda x.\pi \in \interp{\forall x.~ P}$.}

  Assume $\lambda x.\pi \ra^* \lambda x.\pi'$.
  It must be the case that $\pi\ra^*\pi'$, and for any $\t$
  we have $\pi[t/x]\ra^*\pi'[t/x]$ by Proposition~\ref{prop:redsubst}
  and thus $\pi'[t/x]\in\interp{P[t/x]}$ as needed.

\item
  \emph{If $\pi\in\interp{B(\mu B)\t}$,
     then $\mu(B,\t,\pi)\in\interp{\mu B \t}$.}

  From Definition~\ref{def:interp}, we have
  $\interp{\mu B\t} = \lfp{\phi}(\t) = \phi(\interp{\mu 
   B})(\t)$. Using Lemma~\ref{lem:subst}, we obtain that
  $\interp{\mu B \t}
  = \set{\rho\in\SN}{\rho \ra^* \mu(B,\t,\pi') \mbox{ implies } \pi'\in
    \interp{B(\mu B)\t}}$.
  It is now easy to see that
  $\pi\in\interp{B (\mu B) \t}$ implies $\mu(B,\t,\pi)\in\interp{\mu B \t}$:
  for any reduction $\mu(B,\t,\pi) \ra^* \mu(B,\t,\pi')$
  it must be the case that $\pi\ra^*\pi'$ and thus $\pi'\in\interp{B(\mu B)\t}$.

\end{itemize}

Elimination rules also follow a common scheme:
\begin{itemize}

\item
 \emph{If $\pi\in\interp{P\supset Q}$ and $\pi'\in\interp{P}$, then
    $\pi\ap\pi'\in\interp{Q}$.}

 We proceed by induction on the 
 strong normalizability of $\pi$ and $\pi'$.
 By the candidate of reducibility condition on neutral terms, it suffices to 
 show that all immediate reducts $\pi\ap\pi'\ra\pi''$ belong to $\interp{Q}$.
 If $\pi''$ is obtained by a reduction inside $\pi$ or $\pi'$,
 then we conclude by induction hypothesis since the resulting subterm
 still belongs to the expected interpretation.
 Otherwise, it must be that $\pi = \lambda\alpha.\rho$ and the reduct
 is $\rho[\pi'/\alpha]$. In that case we conclude by definition of
 $\pi \in \interp{P \supset Q}$.

\item The cases of $\wedge$, $\vee$ and $\bot$ are treated similarly.

\item
  \emph{If $\pi\in\interp{\forall x.~ P}$, then $\pi\ap t \in \interp{P[t/x]}$.}

  We proceed by induction on the strong normalizability of $\pi$,
  considering all one-step reducts of the neutral term $\pi\ap t$.
  Internal reductions are handled by induction hypothesis.
  If $\pi = \lambda x.\pi'$, our term may reduce at toplevel into
  $\pi'[t/x]$. In that case we conclude by definition of $\interp{\forall x.~ P}$.

\item
  \emph{If $\pi \in \interp{\exists x.~ P}$ and $\pi'\subst{t/x}$ is
    $(\alpha:\interp{P\subst{t/x}}\vdash \interp{Q})$-reducible for any $t$, then
    $\elimex{\pi}{x.\alpha.\pi'} \in \interp{Q}$.}

  We proceed by induction on the strong normalizability of $\pi$ and $\pi'$,
  considering all one-step reducts. The internal reductions are handled
  by induction hypothesis. A toplevel reduction into $\pi'[t/x][\pi''/\alpha]$
  may occur when $\pi = \pair{t,\pi''}$ in which case we have
  $\pi''\in\interp{P[t/x]}$ by hypothesis on $\pi$ and definition
  of $\interp{\exists x.~ P}$. We conclude by hypothesis on $\pi'[t/x]$.

\item
  \emph{If $\pi \in \interp{t\theta = t'\theta}$, $\sigma \in 
    \interp{\Gamma\theta}$ and
    $\pi'_i\theta'$ is
    $(\interp{\Gamma\theta_i\theta'}\vdash\interp{P\theta_i\theta'})$-reducible
    for any $i$ and $\theta'$,
    then $\elimeq{\Gamma,\theta,\sigma,t,t',P,\pi}{(\theta_i.\pi_i)_i} \in \interp{P\theta}$.}

  We proceed by induction on the strong normalizability of the subderivations
  $\pi$, $\sigma$ and $\pi_i$.
  In order to show that a neutral term belongs to a candidate, it suffices to
  consider all its one-step reducts. Reductions occurring inside subterms are
  handled by induction hypothesis. We may also have a toplevel redex when
  $t\theta \equiv t'\theta$ and $\pi = \refl{t\theta}$, reducing to
  $\pi_i\theta'\sigma$ where $\theta'$ is such that
  $\theta_i\theta' \equiv \theta$.
  By hypothesis, $\pi_i\theta'$ is
  $(\interp{\Gamma\theta_i\theta'} \vdash \interp{P\theta_i\theta'})$-reducible
  and $\sigma \in \interp{\Gamma\theta} = \interp{\Gamma\theta_i\theta'}$,
  and thus we have $\pi_i\theta'\sigma\in\interp{P\theta}$ as expected.

\item The case of $\delta_\nu$ is singular, as it follows directly from the definition
  of the interpretation of greatest fixed points. Indeed, we obtain exactly
  $\interp{\nu B \t} = \set{\pi}{\delta_\nu(B,\t,\pi)\in\interp{B(\nu B)\t}}$
  by unfolding the interpretation of greatest fixed points like we did for the 
  least fixed point case above,
  using Definition~\ref{def:interp} and Lemma~\ref{lem:subst}.

\end{itemize}

\subsection{
  Proof of Lemmas~\ref{lem:adeqfun}, \ref{lem:adeqmu} and \ref{lem:adeqnu}}

Let us first introduce the following notation for conciseness:
we say that $\pi$ is $(\vec{x},\X\x\vdash\Y\x)$-reducible
when $\pi\subst{\t/\x}$ is $(\X\t\vdash \Y\t)$-reducible for any $\t$.

We prove the three lemmas simultaneously, generalized as follows
for a predicate operator $B$ of second-order arity\footnote{
  In (1) and (2), $B$ has type $o^{n+1} \ra o$.
  In (3) and (4) we are considering $B$ of type
  $o^n \ra (\vec{\gamma}\ra o) \ra (\vec{\gamma} \ra o)$.
} $n+1$,
predicates $\vec{A}$ and predicate candidates $\vec{Z}$:
\begin{itemize}
\item[(1)] For any $(\x,\X\x\vdash\X'\x)$-reducible $\pi$,
  $\F{B\vec{A}}^+(\x.\alpha.\pi) \in
    \interp{B\vec{Z}\X\supset B\vec{Z}\X'}$.
\item[(2)] For any $(\x,\X\x\vdash\X'\x)$-reducible $\pi$,
  $\F{B\vec{A}}^-(\x.\alpha.\pi) \in
    \interp{B\vec{Z}\X' \supset B\vec{Z}\X}$.
\item[(3)] For any $(\x,\interp{B}\vec{Z}\X\x \vdash \X\x)$-reducible $\pi$,
  $\delta_\mu(\beta,\x.\alpha.\pi)$ is
  $(\interp{\mu (B\vec{Z}) \t} \vdash \X\t)$-reducible.
\item[(4)] For any $(\x,\X\x\vdash \interp{B}\vec{Z}\X\x)$-reducible $\pi$,
  $\nu(\beta,\x.\alpha.\pi)$ is
  $(\X\t \vdash \interp{\nu (B\vec{Z}) \t})$-reducible.
\end{itemize}
We proceed by induction on the number of logical connectives in $B$.
The purpose of the generalization is to keep formulas $\vec{A}$ out of
the picture: those are potentially large but are treated atomically in the 
definition of functoriality, moreover they will be interpreted by candidates
$\vec{Z}$ which may not be interpretations of formulas.
We first prove (3) and (4) by relying on smaller instances of (1),
then we show (1) and (2) by relying on smaller instances of all four 
properties but also instances of (3) and (4) for an operator of the same size.
\begin{itemize}
\item[(1)] We proceed by case analysis on $B$.
  When $B = \lambda \vec{p} \lambda q. q \t$, we have to establish that
  $\F{B\vec{A}}^+(\x.\alpha.\pi) = \lambda \beta. \pi[\t/\x][\beta/\alpha]
  \in \interp{P \t \supset P \t'}$. It simply follows from 
  Lemma~\ref{lem:adequacies} and the hypothesis
  on $\pi$. When $B = \lambda \vec{p} \lambda q. B'\vec{p}$ where $q$ does not
  occur in $B'$, we have to show
  $\F{B\vec{A}}^+(\x.\alpha.\pi) = \lambda \beta. \beta
  \in \interp{B'\vec{Z} \supset B'\vec{Z}}$, which is trivial.

  In all other cases, we use the adequacy properties and conclude by
  induction hypothesis. Most cases are straightforward,
  relying on the adequacy properties.
  In the implication case, \ie\ $B$ is $B_1 \supset B_2$,
  we use induction hypothesis (2) on $B_1$ and (1) on $B_2$.
\ignore{
  Consider for example
  $\F{B_1\vec{A}\wedge B_2\vec{A}}^{+}(\x.\alpha.\pi) =
    \lambda\beta. \pair{
                 \F{B_1\vec{A}}^{+}(\x.\alpha.\pi)(\fst{\alpha}),
                 \F{B_2\vec{A}}^{+}(\x.\alpha.\pi)(\snd{\alpha})}$.
  By induction hypothesis, $\F{B_i\vec{A}}^{+}(\x.\alpha.\pi) \in
  \interp{B_i\vec{Z}\X \supset B_i\vec{Z}\X'}$.
  Now, $\fst{\alpha}$ is
  $(\alpha:\interp{B_1\vec{Z}\X \wedge B_2\vec{Z}\X} \vdash
    \interp{B_1\vec{Z}X})$-reducible, and
  $\F{B_1\vec{A}}^{+}(\x.\alpha.\pi)(\fst{\alpha})$ is thus
  $(\alpha:\interp{B_1\vec{Z}\X \wedge B_2\vec{Z}\X} \vdash
    \interp{B_1\vec{Z}\X'})$-reducible.
  The same holds for the other component of the pair, so the whole pair is
  $(\alpha:\interp{B_1\vec{Z}\X \wedge B_2\vec{Z}\X} \vdash
    \interp{B_1\vec{Z}\X'\wedge B_2\vec{Z}\X'})$-reducible, and
  $\F{B\vec{A}}^+(\x.\alpha.\pi) \in \interp{B\vec{Z}\X \supset B\vec{Z}\X'}$.
}
  Let us only detail the least fixed point case:
  \[ \begin{array}{l}
     \F{\lambda q. \mu (B \vec{A} q) \t}^{+}(\x.\alpha.\pi) \eqdef
     \\ \quad
     \lambda\beta.~ \delta_\mu(\beta, \x.\gamma. \mu (B \vec{A} P',\x,
     \F{\lambda q. B \vec{A} q (\mu (B \vec{A} P')) \x}^{+}(\x.\alpha.\pi) \gamma))
  \end{array} \]
  By induction hypothesis (1) with
  $B := \lambda \vec{p} \lambda p_{n+1} \lambda q. B \vec{p} q p_{n+1} \x$,
  $A_{n+1} := \mu (B \vec{A} P')$
  and $Z_{n+1} := \interp{\mu (B \vec{Z} \X')}$, we have:
  \[ \F{\ldots}^{+}(\x.\alpha.\pi)
  \in \interp{B \vec{Z} \X (\mu (B \vec{Z} \X')) \x \supset
              B \vec{Z} \X' (\mu (B \vec{Z} \X')) \x} \]
  We can now apply the $\supset$-elimination and
  $\mu$-introduction principles to obtain that
  $\mu (B \vec{A} P',\x,
      (\F{\ldots}^{+}(\x.\alpha.\pi))
     \gamma)$
  is $(\gamma:\interp{B\vec{Z}\X(\mu(B\vec{Z}\X'))\x}\vdash
       \interp{\mu(B\vec{Z}\X')\x})$-reducible.
  Finally, we conclude using induction hypothesis (3) with
  $B := \lambda\vec{p}\lambda p_{n+1}\lambda q\lambda \x.~
    B\vec{p} p_{n+1} q \x$,
  $A_{n+1} := P$, $Z_{n+1} := \X$
  and $\X := \interp{\mu (B \vec{Z} \X')}$:
  $ 
     \F{\lambda q. \mu (B \vec{A} q) \t}^{+}
  $ is $(\interp{\mu (B \vec{Z} \X) \t} \vdash
                  \interp{\mu (B \vec{Z} \X') \t})$-reducible.

\item[(2)] Antimonotonicity: symmetric of monotonicity, without the
  variable case.
\item[(3)] Induction:
  we seek to establish that
  $\delta_\mu(\rho,\x.\alpha.\pi) \in \X\t$ when
  $\rho\in\interp{\mu (B\vec{Z}) \t}$
  and $\pi$ is $(\x,\interp{B}\vec{Z}\X\x \vdash \X\x)$-reducible.
  We shall show that $\interp{\mu (B\vec{Z}) \t}$ is included
  in the set of proofs for which this holds, by showing
  that (a) this set is a candidate and (b) it is a prefixed point
  of $\phi$ such that $\interp{\mu (B\vec{Z}) \t} = \lfp{\phi}$.
  Let us consider
   \[ \Y := \t \mapsto \set{\rho}{
         \delta_\mu(\rho,\x.\alpha.\pi) \in \X\t } \]

  First, $\Y\t$ is a candidate for any $\t$:
  conditions (1) and (2) are inherited
  from $\X\t$, only condition (3) is non-trivial.
  Assuming that every one-step reduct of a neutral derivation
  $\rho$ belongs to $\Y$, we prove
  $\delta_\mu(\rho,\x.\alpha.\pi)\in\X\t$.
  This is done by induction on the strong normalizability of $\pi$.
  Using condition (3) on $\X\t$, it suffices to consider one-step 
  reducts: if the reduction takes place
  in $\rho$ we conclude by hypothesis; if it takes place in $\pi$ we
  conclude by induction hypothesis; finally, it cannot take place
  at toplevel because $\rho$ is neutral.

  We now establish that $\phi(\Y) \subseteq \Y$: assuming $\rho\in\phi(\Y)\t$,
  we show that $\delta_\mu(\rho,\x.\alpha.\pi)\in\X\t$. This is done
  by induction on the strong normalizability of $\rho$ and $\pi$,
  and it suffices to show that each one step reduct belongs to $\X\t$,
  with internal reductions handled simply by induction hypothesis.
  Therefore we consider the case where $\rho = \mu(B\vec{A},\t,\pi')$ and
  our derivation reduces to
  $\pi[\t/\x][\F{\lambda q. B\vec{A}q\t}(
      \x.\beta.\delta_\mu(\beta,\x.\alpha.\pi))\pi'/\alpha]$.
  Now, recall that
  $\pi[\t/\x]$ is $(\interp{B\vec{Z}}\X\t \vdash \X\t)$-reducible.
  Since $\mu(B\vec{A},\t,\pi')=\rho\in\phi(\Y)\t$,
  we also have $\pi'\in \interp{B\vec{Z}\Y\t}$.
  By induction hypothesis (1) we obtain that
  $\F{\lambda q.B\vec{A}q\t}(\x.\beta.\delta_\mu(\beta,\x.\alpha.\pi))$
  is $(\interp{B\vec{Z}\Y\t}\vdash \interp{B\vec{Z}\X\t})$-reducible,
  since $\delta_\mu(\beta,\x.\alpha.\pi)$ is $(\x, \beta:\Y\x \vdash \X\x)$-reducible
  by definition of $\Y$.
  We conclude by composing all that.

\item[(4)] Coinduction is similar to induction.
  Let us consider 
  \[ \begin{array}{l} \Y := \t \mapsto \set{\pi\in\SN}{
     \pi \ra^* \nu(\rho,\x.\alpha.\pi)
     \text { implies } \\ \quad
      \rho\in \X\t
      \mbox{ and }
      \pi \mbox{ is ($\x,\alpha:\X\x\vdash \interp{B}\vec{Z}\X\x$)-reducible}}
  \end{array} \]
  It is easy to show that $\Y$ is a predicate candidate,
  and if we show that $\Y \subseteq \interp{\nu (B\vec{Z})}$ we can conclude
  because the properties on $\rho$ and $\pi$ are preserved by reduction.

  We have $\interp{\nu (B\vec{Z})} = \gfp{\phi}$, so it suffices to establish that $\Y$ 
  is a post-fixed point of $\phi$, or in other words that for any $\t$ and 
  $\pi\in\Y\t$, $\delta_\nu(B\vec{A},\t,\pi)\in \interp{B}\vec{Z}\Y\t$.
  We do this as usual by induction on the strong normalizability of $\pi$ and
  the only interesting case to consider is the toplevel reduction,
  which can occur when $\pi = \nu(\rho,\x.\alpha.\pi')$.
  The reduct is
    $\F{\lambda p. B\vec{A} p \t}(\x.\beta.\nu(\beta,\x.\alpha.\pi'))
     \app(\pi'\subst{\t/\x}\subst{\rho/\alpha})$.
  It does belong to $\interp{B\vec{Z}\Y\t}$ because:
  $\rho\in\X\t$ by definition of $\pi\in\Y\t$;
  $\pi'\subst{\t/\x}$ is $(\alpha:\X\t\vdash 
    \interp{B}\vec{Z}\X\t)$-reducible for the same reason;
  and finally
  $\F{\lambda p.~ B\vec{A} p \t}(\x.\beta.\nu(\beta,\x.\alpha.\pi'))
   \in \interp{B\vec{Z}\X\t\supset B\vec{Z}\Y\t}$ by (1) since
  $\nu(\alpha,\x.\alpha.\pi')$ is $(\x;\alpha:\X\x\vdash \Y\x)$-reducible
  by definition of $\Y$.
\end{itemize}

\subsection{Proof of Theorem~\ref{th:adequacy}}

We proceed by induction on the height of $\pi$.
If $\pi$ is a variable, then $\pi\sigma = \sigma(\alpha)$.
Thus, it belongs to $\interp{\Gamma(\alpha)}$ by hypothesis,
and since we are considering a pre-model of the congruence,
and $P \equiv \Gamma(\alpha)$, we have $\pi\sigma\in\interp{P}$.

Other cases follow from the adequacy properties established
previously.
For instance,
if $\pi$ is of the form $\lambda \alpha. \pi'$, then
$P \equiv P_1 \supset P_2$ and $\interp{P} = \interp{P_1 \supset P_2}$.
By induction hypothesis,
$\pi'$ is $(\Gamma,\alpha:\interp{P_1} \vdash \interp{P_2})$-reducible.
Equivalently, $\pi'\sigma$ is $(\interp{P_1}\vdash\interp{P_2})$-reducible,
and we conclude using Lemma~\ref{lem:adequacies}.
In the case where $\pi = \lambda x.\pi'$,
we need to establish that each $\pi'\sigma[t/x]$ belongs to $\interp{P'[t/x]}$.
We obtain this by induction hypothesis,
since $\pi'[t/x]$ has the same height as $\pi'$, which is smaller than $\pi$, 
and we do have $\Gamma[t/x] \vdash \pi'[t/x] : P'[t/x]$.
Similarly, when $\pi = 
\elimeq{\Gamma',\theta',\sigma',t,t',P',\pi'}{(\theta_i.\pi_i)_i}$,
we establish $\pi\sigma \in \interp{P'\theta'}$
by using the induction hypothesis to obtain that
$\sigma'\sigma \in \interp{\Gamma'\theta'}$,
$\pi'\sigma \in \interp{t\theta' = \t'\theta'}$ and,
for any $i$ and $\theta''$,
$\pi_i\theta''$ is
$(\interp{\Gamma'\theta_i\theta''}\vdash \interp{P'\theta_i\theta''})$-reducible.


\subsection{Proof of Theorem~\ref{th:rec}}

We define $\equiv_{a_i\vec{t}}$ (resp. $\equiv_{\prec a_i\vec{t}}$)
to be the congruence resulting
from the extension of $\equiv$ with rule instances $a_j\vec{t'} \rew B$
for $a_j\vec{t'}\preceq a_i\vec{t}$ (resp. $a_j\vec{t'}\prec a_i\vec{t}$).
Let us also write $P \preceq a_i\vec{t}$ (resp. $P \prec a_i\vec{t}$)
when $a_j\vec{t'}\preceq a_i\vec{t}$ (resp. $a_j\vec{t'}\prec a_i\vec{t}$)
for any $a_j\vec{t'}$ which may occur in $P$.
We shall build a family of pre-models $\M^{a_i\vec{t}}$
such that:
\begin{itemize}
\item[(a)] for any $a_i \vec{t} \prec a_j \vec{t'}$,
   $\interp{a_j \vec{t'}}_{\M^{a_i\vec{t}}} = \SN$;
\item[(b)] for any $P \preceq a_j\vec{t'}$ and $a_j\vec{t'}\prec a_i\vec{t}$,
  $\interp{P}_{\M^{a_j\vec{t'}}} = \interp{P}_{\M^{a_i\vec{t}}}$;
\item[(c)] $\M^{a_i\vec{t}}$ is a pre-model of $\equiv_{a_i\vec{t}}$.
\end{itemize}
We proceed by well-founded induction.
Assuming that $\M^{a_j\vec{t'}}$ is defined for all
$a_j\vec{t'} \prec a_i\vec{t}$, we shall thus build
$\M^{a_i\vec{t}}$.

We first define $\M^{\prec a_i\vec{t}}$ by taking each $\hat{a_j}\vec{t'}$
to be the same as in $\M^{a_j\vec{t'}}$ when $a_j\vec{t'} \prec a_i\vec{t}$
and $\SN$ otherwise.
By this definition and property (b) of our pre-models, we have
\[ \interp{P}_{\M^{\prec a_i\vec{t}}} = \interp{P}_{\M^{a_j\vec{t'}}}
\mbox{ for any $P \preceq a_j\vec{t'}$ and $a_j\vec{t'}\prec a_i\vec{t}$.} \]
Next, we observe that $\M^{\prec a_i\vec{t}}$ is a pre-model of
$\equiv_{\prec a_i\vec{t}}$.
It suffices to check it separately for each rewrite rule.
An instance $P \rew Q$ of a rule defining the initial congruence cannot 
involve the new predicates, so in that case we do have
\[ \interp{P}_{\M^{\prec a_i\vec{t}}} = \interp{P}_{\M} = 
\interp{Q}_{\M} = \interp{Q}_{\M^{\prec a_i\vec{t}}}. \] For a rule instance
$a_j\vec{t'} \rew B$ with $a_j\vec{t'} \prec a_i\vec{t}$,
the property is similarly inherited from $\M_{a_j\vec{t'}}$
because $B \preceq a_j\vec{t'}$ by (2):
\[ \interp{a_j\vec{t'}}_{\M^{\prec a_i\vec{t}}} =
   \interp{a_j\vec{t'}}_{\M^{a_j\vec{t'}}} = 
   \interp{B}_{\M^{a_j\vec{t'}}} = 
   \interp{B}_{\M^{\prec a_i\vec{t}}} \]

We finally build $\M^{a_i\vec{t}}$ to be the same as $\M^{\prec a_i\vec{t}}$
except for $\hat{a_i}\vec{t}$ which is defined as follows:
\begin{itemize}
\item If there is no rule $a_i\vec{t''} \rew B$ such that
  $\vec{t''\theta} \equiv \vec{t}$, we define $\hat{a_i}\vec{t}$ to be $\SN$.
\item Otherwise, pick any such $B$, and
  define $\hat{a_i}\vec{t}$ to be $\interp{B\theta}_{\M^{\prec a_i\vec{t}}}$.
  This is uniquely defined: for any other $a_i\vec{t'} \rew B'$
  such that $a_i\vec{t} = (a_i\vec{t'})\theta'$,
  we have $B\theta \equiv B'\theta'$ by (1), and thus
  $\interp{B'\theta'}_{\M^{\prec a_i\vec{t}}} = 
   \interp{B\theta}_{\M^{\prec a_i\vec{t}}}$ since
  $\M^{\prec a_i\vec{t}}$ is \emph{a fortiori} a pre-model of $\equiv$.
\end{itemize}
This extended pre-model satisfies (a) by construction.
It is also simple to show that it satifies (b).
To check that it verifies (c) we check separately each instance of a rewrite 
rule: by construction, our pre-model is compatible with instances of
the form $a_i\vec{t}\rew B$, and it inherits that property from $M^{\prec 
a_i\vec{t}}$ for other instances.

Finally, we define our new pre-model $\M'$ by taking each $\hat{a_i}\vec{t}$
in $\M^{a_i\vec{t}}$. It is a pre-model of the extended congruence: it is
easy to check that it is compatible with all rewrite rules.

%% file: formalization.tex
\section*{Appendix II: Formalization of Strong Normalizability} 
\setcounter{subsection}{0}

We detail below the formalization of Tait's strong normalizability argument 
described in Section~\ref{sec:rec-defs}.
The full Abella scripts are available at
\url{http://www.lix.polytechnique.fr/~dbaelde/lics12}.

Following the two-level reasoning methodology facilitated by Abella,
we first define the objects and judgments of interest in a module
file shown on Figure~\ref{fig:snmod}. The specification is given
by means of hereditary Harrop clauses\footnote{
  These clauses also define a $\lambda$Prolog, which gives a way
  to execute them directly.
}. Adequate representations
are obtained by considering uniform proofs for the corresponding clauses.
For instance, uniform proofs of $\Gamma \vdash of~M~T$ are in bijection
with typing derivations in simply typed $\lambda$-calculus.
Abella allows one to reason over the specified objects through 
this representation methodology. Derivability is inductively defined
as a builtin predicate in Abella, written in a concise notation:
\verb#{C |- of M T}# corresponds to the derivability of
$C \vdash of~M~T$.
More details on the methodology and syntax of Abella, refer
to \url{http://abella.cs.umn.edu}.

\begin{figure*}
\begin{verbatim}
isty iota.                                       step (app (abs M) N) (M N).
isty (arrow T T') :- isty T, isty T'.            step (app M N) (app M' N) :- step M M'.
                                                 step (app M N) (app M N') :- step N N'.
istm (app M N) :- istm M, istm N.                step (abs M) (abs M') :- pi x\ step (M x) (M' x).
istm (abs M) :- pi x\ istm x => istm (M x).
                                                 steps M M.
of (app M N) T' :- of N T, of M (arrow T T').    steps M N :- step M M', steps M' N.
of (abs M) (arrow T T') :-
  isty T, pi x\ of x T => of (M x) T'.           subst (app M N) (app M' N') :-
                                                   subst M M', subst N N'.
                                                 subst (abs M) (abs M') :-
                                                   pi x\ pi y\ subst x y => subst (M x) (M' y).
\end{verbatim}
\caption{Module file for the Abella formalization} \label{fig:snmod}
\end{figure*}

\subsection{Preliminaries}

We first prove that \verb#steps# is transitive and that it is a congruence.

\begin{verbatim}
Theorem steps_steps : forall M N P,
  {steps M N} -> {steps N P} ->
  {steps M P}.

Theorem steps_app_left : forall M M' N,
  {steps M M'} -> {steps (app M N) (app M' N)}.

Theorem steps_app_right : forall M M' N,
  {steps M M'} -> {steps (app N M) (app N M')}.

Theorem steps_app : forall M M' N N',
  {steps M M'} -> {steps N N'} ->
  {steps (app M N) (app M' N')}.

Theorem steps_abs : forall M M', nabla x,
  {steps (M x) (M' x)} ->
  {steps (abs M) (abs M')}.
\end{verbatim}

Next, we define open terms, which cannot be done at the specification level
like, for example, the definition of \verb#isty#. We prove a few basic
properties of open terms.

\begin{verbatim}
Define isotm : tm -> prop by
  nabla x, isotm x ;
  isotm (app M N) := isotm M /\ isotm N ;
  isotm (abs M) := nabla x, isotm (M x).

Theorem isotm_subst : forall M N, nabla x,
  isotm (M x) -> isotm N -> isotm (M N).

Theorem isotm_step : forall M M',
  isotm M -> {step M M'} -> isotm M'.

Theorem step_osubst_steps : forall M N N', nabla x,
  isotm (M x) -> {step N N'} -> {steps (M N) (M N')}.
\end{verbatim}

\subsection{Strong normalizability}

We define strong normalizability and prove some basic properties about it.

\begin{verbatim}
Define sn : tm -> prop by
  sn M := forall N, {step M N} -> sn N.

Theorem var_sn : nabla x, sn x.

Theorem sn_step_sn : forall M N,
  sn M -> {step M N} -> sn N.

Theorem sn_preserve : forall M, nabla x,
  sn (app M x) -> sn M.

Theorem sn_app : forall M N,
  sn M -> sn N ->
  (forall M', {steps M (abs M')} -> false) ->
  sn (app M N).
\end{verbatim}

\subsection{Reducibility}

We now give the definition of reducibility. Abella issues a warning here,
because the definition is not monotone, and is thus not formally supported by
its underlying theory. However, as explained in Section~\ref{sec:rec-defs},
this recursive definition can be justified as rewrite rules in our 
framework. Below, it is always going to be used following this
interpretation.

\begin{verbatim}
Define red : tm -> ty -> prop by
  red M iota := sn M ;
  red M (arrow T T') := forall N,
    isotm N -> red N T -> red (app M N) T'.
\end{verbatim}

We now prove the three conditions defining candidates of reducibility.
We first show that reducible terms are SN, and simultaneously that
variables are reducible, which requires a generalization to
showing that $x \ N_1 \ldots N_k$ is reducible when the $N_i$ are SN.

\begin{verbatim}
Define vargen : tm -> prop by
  nabla x, vargen x ;
  vargen (app M N) := vargen M /\ sn N.

Theorem vargen_step_vargen : forall M N,
  vargen M -> {step M N} -> vargen N.

Theorem vargen_steps_noabs : forall M M',
  vargen M -> {steps M (abs M')} -> false.

Theorem vargen_sn : forall M, vargen M -> sn M.

Theorem red_sn_gen : forall M T,
  {isty T} ->
  (red M T -> sn M) /\ (vargen M -> red M T).

Theorem var_red : forall T, nabla x,
  {isty T} -> red x T.

Theorem red_sn : forall M T,
  {isty T} -> red M T -> sn M.
\end{verbatim}

The second condition is that reducts remain in reducibility sets.

\begin{verbatim}
Theorem red_step : forall M M' T,
  {isty T} -> red M T -> {step M M'} ->
  red M' T.

Theorem red_steps : forall M M' T,
  {isty T} -> red M T -> {steps M M'} ->
  red M' T.
\end{verbatim}

Finally,
if all one-step reducts of a neutral term are in a set, then so is the 
term. We only prove it for neutral terms which are applications.
Here, the inner induction is taken care of using an auxiliary lemma.

\begin{verbatim}
Theorem cr3_aux : forall M1 M2 N T1 T',
  {isty T1} -> sn N -> isotm N -> red N T1 ->
  (forall M1 M2,
    (forall M',
       {step (app M1 M2) M'} -> red M' T') ->
    red (app M1 M2) T') ->
  (forall M',
     {step (app M1 M2) M'} ->
     red M' (arrow T1 T')) ->
  red (app (app M1 M2) N) T'.

Theorem red_anti : forall M N T,
  {isty T} ->
  (forall M',
     {step (app M N) M'} -> red M' T) ->
  red (app M N) T.
\end{verbatim}

\subsection{Contexts}

We characterize the contexts used in derivations of \verb.of M T.
and \verb.subst M T. that are involved in the proof of adequacy.
We also define separately their relationship, using \verb#mapctx#.
This approach requires a fair number of book-keeping lemmas.

\begin{verbatim}
Define name : tm -> prop by nabla x, name x.

Define ofctx : olist -> prop by
  ofctx nil ;
  nabla x, ofctx (of x T :: C) :=
    {isty T} /\ ofctx C.

Define substctx : olist -> prop by
  substctx nil ;
  nabla x, substctx (subst x (M x) :: C) :=
    nabla x, isotm (M x) /\ substctx C.

Define mapctx : olist -> olist -> prop by
  mapctx nil nil ;
  nabla x,
    mapctx (of x T :: C) (subst x (M x) :: C')
  :=
      nabla x,
        {isty T} /\ isotm (M x) /\
        red (M x) T /\ mapctx C C'.

Theorem ofctx_member_isty : forall C T,
  ofctx C -> member (isty T) C -> false.

Theorem isty_weaken : forall C T,
  ofctx C -> {C |- isty T} -> {isty T}.

Theorem ofctx_member_isty : forall C M T,
  ofctx C -> member (of M T) C -> {isty T}.

Theorem of_isty : forall C M T,
  ofctx C -> {C |- of M T} -> {isty T}.

Theorem mapctx_of : forall G G' M T,
  mapctx G G' -> member (of M T) G ->
  name M /\
  exists M', red M' T /\ member (subst M M') G'.

Theorem mapctx_subst : forall G G' M M',
  mapctx G G' -> member (subst M M') G' ->
  name M /\
  exists T, red M' T /\ member (of M T) G.

Theorem mapctx_split : forall C C',
  mapctx C C' -> ofctx C /\ substctx C'.

Theorem ofctx_member_name : forall C M T,
  ofctx C -> member (of M T) C -> name M.

Theorem of_isotm : forall C M T,
  ofctx C -> { C |- of M T } -> isotm M.

Theorem substctx_member : forall C M M',
  substctx C -> member (subst M M') C ->
  name M /\ isotm M'.

Theorem subst_isotm : forall C M M',
  substctx C -> isotm M -> { C |- subst M M' } ->
  isotm M'.

Theorem member_not_fresh :
  forall X L, nabla (n:tm),
    member (X n) L -> exists X', X = n\X'.

Theorem substctx_member_unique_aux :
  forall C M M', nabla x,
    substctx (C x) ->
    member (subst x (M x)) (C x) ->
    member (subst x (M' x)) (C x) ->
    M = M'.

Theorem substctx_member_unique : forall C X M M',
  substctx C ->
  member (subst X M) C ->
  {C |- subst X M'} -> M = M'.
\end{verbatim}

\subsection{Adequacy theorem}

\begin{verbatim}
Theorem abs_case : forall M N T', nabla x,
  isotm (M x) -> {isty T'} ->
  sn (M x) -> sn N ->
  red (M N) T' ->
  red (app (abs M) N) T'.

Theorem of_red : forall M M' T C C',
  mapctx C C' ->
  {  C |- of M T } ->
  { C' |- subst M M' } ->
  red M' T.
\end{verbatim}

To apply the adequacy result and obtain strong normalizability,
it only remains to show that for any typed term we can define
the identity substitution with which we have \verb#{ C' |- subst M M }#,
from which \verb#red M T# and \verb#sn M# follow.